\documentclass[
  10pt,
  notitlepage,
  twocolumn,
  pra,
  aps,
  longbibliography,
  preprintnumbers,
  superscriptaddress,
  nofootnotebib,
  floatfix
]{revtex4-2}

\pdfoutput=1
\usepackage{printlen}

\usepackage[utf8]{inputenc}
\usepackage[T1]{fontenc}
\usepackage[english]{babel}

\usepackage{times}
\usepackage{dsfont}
\usepackage{amsmath, amsfonts, amsthm}
\usepackage{graphicx}
\usepackage{float}
\usepackage{tabularx}
\usepackage{subcaption}
\usepackage{hyperref}
\usepackage[dvipsnames]{xcolor}
\usepackage[braket, qm]{qcircuit}
\usepackage{tikz}
\usepackage{lipsum}
\usepackage{caption}

\definecolor{lightgrayfill}{gray}{0.97}
\definecolor{darkgrayborder}{gray}{0.25}

\usepackage{amsmath}
\usetikzlibrary{positioning, calc, arrows.meta}

\newcommand{\blochsphere}[6]{%
  \begin{scope}[shift={(#2,#3)}]
    \draw[line width=0.8pt] (0,0) circle (#4);
    \draw[line width=0.6pt] (-#4,0) -- (#4,0); % X
    \draw[line width=0.6pt] (0,-#4) -- (0,#4); % Z
    \pgfmathsetmacro{\th}{#5}
    \pgfmathsetmacro{\R}{#4}
    \pgfmathsetmacro{\vx}{\R*sin(\th)}
    \pgfmathsetmacro{\vz}{\R*cos(\th)}
    \draw[-{Stealth[length=3mm,width=2mm]}, line width=0.7pt] (0,0) -- (\vx,\vz);
    \node[font=\footnotesize] at (0,#4+0.26) {$|1\rangle$};
    \node[font=\footnotesize] at (0,-#4-0.26) {$|0\rangle$};
    \node[font=\scriptsize] at (0,-#4-0.55) {#6};
  \end{scope}
}

\newtheorem{lem}{Lemma}

\def\id{\mathds{1}}
\newcommand{\swap}{\text{SWAP}}
\def\h{\hat H}
\newcounter{sectionnum}
\newcommand{\sectionnum}{\stepcounter{sectionnum}\thesectionnum}

\newcommand{\mg}[1]{{\color{orange}}}

\begin{document}
\title{Double-Bracket Algorithmic Cooling}
%Dynamic Algorithmic Cooling via Double-Bracket Quantum Algorithm \\with Multicolor Drives on Superconducting Qubits}

\author{Mohammed Alghadeer}  
\email{mohammed.alghadeer@physics.ox.ac.uk}
\affiliation{Department of Physics, Clarendon Laboratory, University of Oxford, Oxford, OX1 3PU, United Kingdom}
\author{Khanh Uyen Giang}
\affiliation{School of Physical and Mathematical Sciences, Nanyang Technological University, 637371, Singapore}
\author{Shuxiang Cao}
\affiliation{Department of Physics, Clarendon Laboratory, University of Oxford, Oxford, OX1 3PU, United Kingdom}
\author{Simone D. Fasciati}
\affiliation{Department of Physics, Clarendon Laboratory, University of Oxford, Oxford, OX1 3PU, United Kingdom}
\author{Michele Piscitelli}
\affiliation{Department of Physics, Clarendon Laboratory, University of Oxford, Oxford, OX1 3PU, United Kingdom}
\author{Nelly Ng}
\affiliation{School of Physical and Mathematical Sciences, Nanyang Technological University, 637371, Singapore}
\affiliation{Centre for Quantum Technologies, Nanyang Technological University, 50 Nanyang Avenue, 639798 Singapore}
\author{Peter J. Leek} 
\affiliation{Department of Physics, Clarendon Laboratory, University of Oxford, Oxford, OX1 3PU, United Kingdom}
\author{Marek Gluza}
\affiliation{School of Physical and Mathematical Sciences, Nanyang Technological University, 637371, Singapore}
\author{Mustafa Bakr} 
\email{mustafa.bakr@physics.ox.ac.uk}
\affiliation{Department of Physics, Clarendon Laboratory, University of Oxford, Oxford, OX1 3PU, United Kingdom}

%%%%%%%%%%%%%% Final Abstarct 
\begin{abstract}
Algorithmic cooling shows that it is possible to locally reduce the entropy of a qubit belonging to an isolated ensemble such as nuclear spins in molecules or nitrogen-vacancy centers in diamonds. In the same physical setting, we introduce double-bracket algorithmic cooling (DBAC), a protocol that systematically suppresses quantum coherence of pure states. DBAC achieves this by simulating quantum imaginary-time evolution through recursive unitary synthesis of Riemannian steepest-descent flows and it utilizes density-matrix exponentiation as a subroutine. This subroutine makes DBAC a concrete instance of a dynamic quantum algorithm that operates using quantum information stored in copies of the input states. Thus, the circuits of DBAC are independent of the input state, enabling the extension of algorithmic cooling from targeting entropy to quantum coherence without resorting to measurements. Akin to Nernst’s principle, DBAC increases the cooling performance when including more input qubits which serve as quantum instructions. Our work demonstrates that dynamic quantum algorithms are a promising route toward new protocols for foundational tasks in quantum thermodynamics.
\end{abstract}

\maketitle

%%%%%%%%%%%%%% Final Intro 
\section{Introduction}

A century ago, Walther Nernst pointed out that it is impossible to cool a system perfectly to absolute zero in a finite number of operations~\cite{nernst1912}. At that time, the arrival of quantum technologies could not have been anticipated.
Even though it preceded quantum mechanics, it was in fact Nernst's unattainability principle that prompted Einstein to argue that a quantum theory, rather than a classical theory, is necessary to understand matter at extremely low temperatures~\cite{Einstein1913}. 
More recently, rigorous theoretic approaches showed that Nernst’s law features centrally in the quantum regime as it can be proven from general assumptions~\cite{wilming2017third, masanes2017general}.
Today, the next major area to explore the validity of Nernst's phenomenology is quantum computing.
If quantum computers could cool a simulated many-body system to its ground state in finite time then they would be violating an empirical law of thermodynamics that stood for a century.
However, contemporary computational complexity theory indicates that quantum computers cannot achieve this~\cite{QMA}.
In this sense, quantum computing resembles other physical processes and appears to obey Nernst's principle as well.

This connection between quantum computing and laws of thermodynamics can be explored further, by explicitly specifying cooling steps that are both physically relevant and as general as possible, followed by showing that perfect cooling requires unbounded effort.
Algorithmic cooling (AC)~\cite{schulman05, Raeisi_2015, alhambra2019heat, Park2016} is a paradigmatic framework for this because, as opposed to cooling by dissipating entropy to the environment~\cite{Kraus_2008}, it redistributes entropy within the system to cool a part of it.
In this setting, a fundamental limit to cooling can be derived, and the form of the optimal transformation is known~\cite{schulman05,Raeisi_2015,alhambra2019heat,Park2016}.

In this work, we combine two novel quantum algorithms, density-matrix exponentiation~\cite{lloyd2014quantum,kjaergaard,kimmel,Wei2023hermpreserving} and double-bracket quantum imaginary-time evolution~\cite{dbqite}, to construct a protocol in the spirit of algorithmic cooling, which enables probing the unattainability principle in an extreme regime. 
Specifically, we consider $n$ independent identical qubits which have no entropy but, as is found in NMR and NV-centers, may appear to have non-zero entropy if the observer is incapable of non-demolition measurements. 
Our algorithm, double-bracket algorithmic cooling (DBAC), achieves cooling of one qubit despite this lack of knowledge. 
As opposed to previous algorithmic cooling constructions, our proposal is naturally compatible with typical native interactions of qubits, making quantum compilation tailored to the quantum architecture.
The protocol is dynamic, namely, its operation is programmed at runtime by the encountered state of the system. 
Perfect operation requires infinitely many instructions, which effectively recovers Nernst's statement that perfect cooling would require infinite resources. 
We present the first experimental implementation of a dynamic quantum algorithm applied to the fundamental thermodynamic task of cooling. 

This paper is organized as follows. 
In Section~\sectionnum, we introduce the DBAC protocol. 
In Section~\sectionnum, we present experimental results. 
In Section~\sectionnum, we discuss our findings  and outlook.

% --------- DBAC schematic  -----------

\begin{figure*}[t]
\centering
\resizebox{1\textwidth}{!}{% scales whole figure
\begin{tikzpicture}
  % === Outer box node ===
  \node[fill=white, draw=black!50, rounded corners=6pt,
        line width=1pt, inner sep=6pt] (outer) {%
    \begin{tikzpicture}[
      font=\rmfamily,
      >=Stealth,
      scale=0.8, every node/.style={scale=0.8},
      box/.style    ={draw,rounded corners=2pt,minimum width=35mm,minimum height=10mm,
                      line width=0.8pt,align=center,inner sep=3pt,fill=blue!10},
      thinbox/.style={draw,rounded corners=2pt,minimum width=45mm,minimum height=9mm,
                      line width=0.7pt,align=center,inner sep=3pt,fill=blue!5},
      infobox/.style={draw,rounded corners=3pt,minimum width=40mm,minimum height=12mm,
                      line width=0.8pt,align=center,inner sep=4pt,font=\footnotesize}
    ]

    % Define Bloch sphere command
    % \newcommand{\blochsphere}[6]{
    %   \draw[thick] (#2,#3) circle (#4);
    %   \draw[thick] (#2,#3) ellipse ({#4} and {0.3*#4});
    %   \draw[thick] ({#2-#4*sin(20)},{#3-0.3*#4*cos(20)}) arc (200:340:{#4*sin(20)} and {0.3*#4*cos(20)});
    %   \draw[->,red,thick] (#2,#3) -- ++({#4*0.8*cos(#5)},{#4*0.8*sin(#5)});
    %   \node[below,font=\scriptsize] at (#2,{#3-#4-0.3}) {#6};
    %   \node[right,font=\scriptsize] at ({#2+#4+0.2},#3) {$|1\rangle$};
    %   \node[left,font=\scriptsize] at ({#2-#4-0.2},#3) {$|0\rangle$};
    % }

    \providecommand{\blochfill}{none} % default: no fill
    
    % Define Bloch sphere command (6 args, unchanged interface) + fill
    \renewcommand{\blochsphere}[6]{%
      \fill[\blochfill] (#2,#3) circle (#4); % <-- only addition
      \draw[thick] (#2,#3) circle (#4);
      \draw[thick] (#2,#3) ellipse ({#4} and {0.3*#4});
      % --- Vertical Z axis (|0> to |1>) ---
      % \draw[dashed,gray,thick] (#2-#4,#3) -- (#2+#4,#3);
      \draw[dashed,gray,thick] (#2-#4,#3) -- (#2+#4,#3);
      % \draw[thick] ({#2-#4*sin(20)},{#3-0.3*#4*cos(20)}) arc (200:340:{#4*sin(20)} and {0.3*#4*cos(20)});
      \draw[->,red,thick] (#2,#3) -- ++({#4*0.8*cos(#5)},{#4*0.8*sin(#5)});
      % Arrows between spheres
      \draw[->,line width=0.8pt] (8.8,5.7) -- (8.8,4.3); % from sphere 1 to sphere 2
      \draw[->,line width=0.8pt] (8.8,2.7) -- (8.8,1.3); % from sphere 2 to sphere 3
      \node[below,font=\scriptsize] at (#2,{#3-#4-0.3}) {#6};
      \node[right,font=\scriptsize] at ({#2+#4+0.001},#3) {$|1\rangle$};
      \node[left,font=\scriptsize]  at ({#2-#4-0.001},#3) {$|0\rangle$};
}

    \node[infobox,fill=brown!20] (advantage) at (-1.5,7.7) {Instruction copies of $|\psi\rangle$ drive $U_{\rm DME}$ with $e^{\pm it\hat H}$ echoes; \\ recursion scales resources and increases cooling};

    % ==== KEY ADVANTAGE BOX (gold) ====
    \node[infobox,fill=orange!20] (advantage) at (4.5,7.7) {\textbf{Key Advantage:}\\Direct coherence manipulation\\without measurement dephasing};

    % ==== NERNST CONNECTION BOX (gray) ====
    % purple!15   orange!20   violet!20   magenta!15    
    % green!20    green!30    brown!20   blue!15    gray!15   teal!20     cyan!10     red!20
    \node[infobox,fill=brown!20] (nernst) at (-1,-0.7) {\textbf{Nernst Principle:}\\Perfect cooling ($T=0$) requires\\infinite resources ($M \to \infty$)};

    % ==== LEFT COLUMN (instruction copies) ====
    \node[box] (copies) at (-7.5,5.0) {Instruction copies $|\psi\rangle$};

    % ==== MIDDLE COLUMN (DME sequence) ====
    % purple!15   orange!20   violet!20   magenta!15    
    % green!20    green!30    brown!20   blue!15    gray!15   teal!20     cyan!10     red!20
    
    \node[thinbox,fill=cyan!10] (em)   at (-1.0,5.8) {$e^{-it\hat H}$};
    \node[thinbox,fill=cyan!10] (udme) at (-1.0,4.5) {$U_{\rm DME}=e^{-it(X\!\otimes\!X+Y\!\otimes\!Y+Z\!\otimes\!Z)/2}$};
    \node[thinbox,fill=cyan!10] (ep)   at (-1.0,3.2) {$e^{+it\hat H}$};

    % Resource growth note
    \node[font=\small] at (-1.0,2.25) {Resources grow with steps: $M_1, M_2, \ldots$};

    % ==== RIGHT COLUMN (target qubit) ====
    \node[box,fill=purple!15] (target) at (4.5,4.5) {Target qubit};

    % ==== COMPARISON PANEL (gold) ==== 
    % purple!15   orange!20   violet!20   magenta!15    green!20    green!30    brown!20   blue!15
    \node[infobox,fill=gray!15] (comparison) at (-7.5,7.7) {\textbf{vs. Traditional Methods:}\\HBAC: Requires measurements \\ and a number of entangling gates \\ even for small $n$\\  
    %Requires measurements\\+ 24 CZ gates for $n=3$\\
    DBAC: Native ZZ interactions
    };

    % ==== BOTTOM ROW (recursion loop) ====
    \node[box] (fresh) at (-7.5,1.2) {Input copies of $|\psi\rangle$};
    \node[box,fill=teal!20] (udme2) at (-1.0,1.2) {$U_{\rm DME}$};

    % ==== ARROWS ====
    \draw[->,line width=0.8pt] (copies.east) -- ++(1.,0) |- node[above, font=\small, pos=0.5] {DME step} (em.west);
    \draw[->,line width=0.8pt] (em.south) -- (udme.north);
    \draw[->,line width=0.8pt] (udme.south) -- (ep.north);
    \draw[->,line width=0.8pt] (ep.east) -- (target.west);

    % Recursion arrows
    \draw[->,line width=0.8pt] (fresh.east) -- node[above, font=\small] {iterate (recursion)} (udme2.west);
    \draw[->,line width=0.8pt] (fresh.north) -- (copies.south);

    % ==== BLOCH SPHERES ====
    % \blochsphere{b1}{7.5}{6.5}{0.9}{55}{Before}
    % \blochsphere{b2}{7.5}{3.5}{0.9}{30}{After one step}
    % \blochsphere{b3}{7.5}{0.5}{0.9}{18}{After two steps}
    
    \renewcommand{\blochfill}{blue!15}
    \blochsphere{b1}{7.7}{6.5}{1}{55}{Before}
    
    \renewcommand{\blochfill}{teal!20}
    \blochsphere{b2}{7.7}{3.5}{1}{30}{After one step}
    
    \renewcommand{\blochfill}{purple!15}
    \blochsphere{b3}{7.7}{0.5}{1}{18}{After two steps}
    
    % Optional: reset
    \renewcommand{\blochfill}{none}

    % Final recursion connection
    \draw[->,line width=0.8pt] (udme2.east) -- (3.5,1.2) -- (3.5,4);

    % Highlight text
    \node[font=\footnotesize,color=black!70!black] at (3.9,5.8) {\textbf{No measurements}};
    \node[font=\footnotesize,color=black!70!black] at (3.4,5.5) {\textbf{required!}};

    \end{tikzpicture}%
  };
\end{tikzpicture}}
\caption {Schematic of double-bracket algorithmic cooling (DBAC). The protocol applies the density-matrix exponentiation unitary $U_{\rm DME}$ between a target data qubit and input instruction copies of $|\psi\rangle$. Each application is bracketed by $e^{\pm it\hat H}$ echoes, producing a double-bracket step that reduces the energy of the target qubit. Because the applied unitary operations are state-agnostic, the cooling dynamics are programmed “on the fly” by the instruction copies themselves, without requiring mid-circuit measurements. DBAC thus generalizes algorithmic cooling from entropy reduction to coherent state manipulation, and -consistent with the Nernst unattainability principle- perfect cooling would require infinitely many instruction qubits. The flow is as follows: input copies of $|\psi\rangle$ are promoted to instruction copies, which together with the target qubit undergo the $U_{\rm DME}$ operation bracketed by echoes; the updated target qubit is thereby cooled, and recursion repeats with further instruction copies ($M_1,M_2$,…) as resources scale.
}
\label{fig:dbac}
\end{figure*}
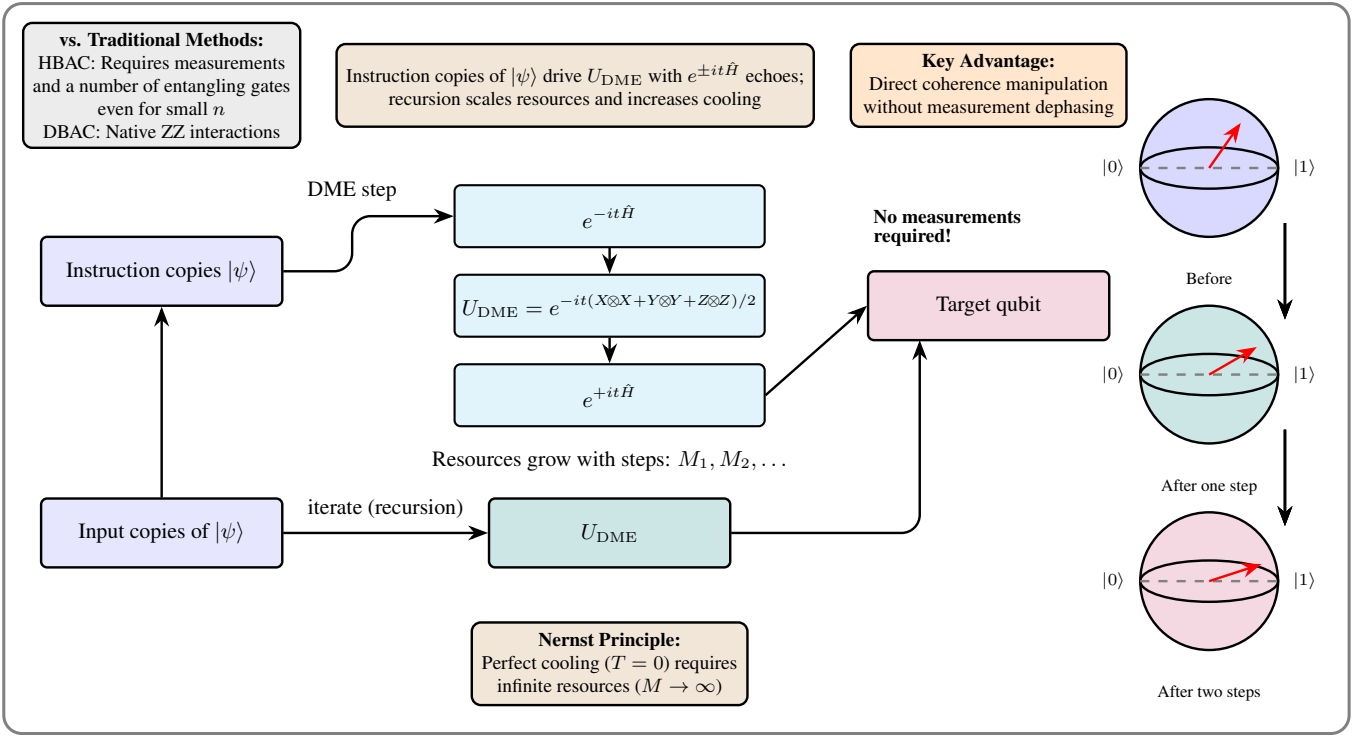

%%%%%%%%%%%%%% Final DBAC  
\section{Double-bracket algorithmic cooling (DBAC)}
The objective of AC is to perform a unitary operations on $n$ qubits which are prepared by thermal contact to the same bath such that one of those qubits will have as little entropy as possible. This local entropy reduction is limited, analogously to Nernst's principle~\cite{Park2016,OverviewAACPhysRevA.110.022215}.
In this work we aim to uncover an even more striking appearance of such limitations in a quantum system.
We will do that by considering a system with no entropy but including quantum superpositions.
Quantum coherence~\cite{streltsov2017colloquium} quantifies the amount of quantum superpositions in a quantum state.
The difficulty in manipulating quantum coherence is that measurements collapse superpositions.
Moreover, a quantum state in a superposition generally appears as a mixed state upon measurement. 
However, we will demonstrate that our quantum algorithm enables a process in which $n$ qubits prepared in the same pure state $\ket{\psi}$ can be used to reduce the quantum coherence of a selected qubit.
This will generalize AC from entropy to quantum coherence. The former can have a classical rather than quantum source while the latter has an exclusively quantum origin.

Our discussion is particularly relevant if measurements are destructive and in every round of an experiment $
\ket \psi$ changes.
In this case if one would average the measurement outcome of each qubit then the repeated qubit state would appear mixed~\cite{Nagasawa_2024}.
Quantum coherence renders standard AC protocols inapplicable as they typically assume states in form of classical mixtures of $\rho = p\ket 0\bra0+(1-p) \ket1\bra1$ and an active dephasing operation may not be available.
Moreover, dephasing would add entropy making cooling to a perfect pure state theoretically impossible for AC~\cite{Park2016,OverviewAACPhysRevA.110.022215}.
In contrast, a single-qubit rotation could potentially rotate any of the qubits to a perfect ground state $\ket\psi\rightarrow \ket 0$.
An immediate idea for that is to perform tomography on the qubits and compile this rotation using information gained from measurements.
However, Nernst's phenomenology continues to hold: even with perfect measurements, infinitely many copies of the state would be needed to find the exact parameters for the operation.
This makes the tomographic approach inapplicable for systems with few qubits and those without fast feedback loops.

We tackle this problem by introducing an algorithm that can perform an operation analogous to a single qubit rotation to ground state $\ket 0$ but without resorting to tomography. 
The basic idea is that imaginary-time evolution (ITE) defined as
\begin{equation} \label{eq:ite}
    \ket{\psi(\tau)} =  \frac{e^{-\tau \h}\ket{\psi} }{\|e^{-\tau \h}\ket{\psi}\|}\
    ,
\end{equation}
maps any pure initial state to another pure state and, importantly the resulting state $\ket{\psi(\tau)}$ becomes increasingly closer to the ground-state of the Hamiltonian $\h$, converging for large ITE durations $\tau\rightarrow\infty$.
This is true largely irrespective of the input state, as long as it is not a perfect excited eigenstate.
Our objective is similar, as we want to rotate the qubit to $\ket 0$ in computational basis without knowing the input state and so we will choose the Hamiltonian in ITE to be $\h=-Z$ where $Z$ is the Pauli-$Z$ matrix. 

Here, we present double-bracket algorithmic cooling (DBAC) algorithm, that uses recent results in quantum computing for implementing ITE obliviously, i.e. without knowledge about the encountered quantum information~\cite{gluza2024double,qdp}.
To achieve that, similarly to AC, we consider $n$ copies of $\ket \psi$ and, unlike AC, DBAC processes the encountered quantum states coherently to avoid the need for, otherwise, mid-circuit measurements.
Requiring $n>1$ copies is in line with a Nernst-like limitation, i.e. DBAC converges to the perfect ground state only in the limit $n\rightarrow\infty$.
We next explain the DBAC procedure whose ingredients are illustrated in Fig.~\ref{fig:dbac}.

Since the ITE state is well normalized for any $\tau$, the state at time $\tau$ can, in principle, be obtained from the initial state $\ket{\psi}$ through a unitary. 
As shown in Refs.~\cite{dbqite,dbqsp,dbgrover}, this unitary takes a double-bracket form $U(\tau)= e^{t(\tau)[\ket{\psi}\bra{\psi},\h]}$, that is, it has the form of an exponential of a commutator. In the corresponding Heisenberg equations of motion this leads to a double commutator, which motivates the term double-bracket algorithmic cooling (DBAC).
Here $t(\tau)$ is a rescaling of ITE duration, and unlike $\tau$, it can be interpreted as the duration of a physically occurring evolution, since $\hat W = i[\ket{\psi}\bra{\psi},\h]$ formally defines a valid Hamiltonian.
While $U(\tau)$ cannot be directly implemented on quantum computers, it can be approximated using so-called group commutator formulas~\cite{gluza2024double}.
For qubits, this approximation can lead to the familiar structure of Grover’s search algorithm~\cite{dbgrover}. 
Using the group commutator we define the first DBAC step as
\begin{align}\label{eq: GCI state main}
\ket{\psi_\text{DBAC}} =e^{it\h} e^{it\ket{\psi}\bra{\psi}}   
    e^{-it\h}\ket{\psi}\ .
\end{align}
Using mathematical bounds for small $\tau$, we have $\ket{\psi(\tau)} \approx \ket{\psi_\text{DBAC}}$~\cite{dbqite,dbgrover,gluza2024double}.
For arbitrary $\tau$, and since we are dealing with a single qubit, a stronger result holds: the initial energy is given by
\begin{align}    
E_\text{DBAC} = E_0-2\sin^2(t)(1-E_0^2)\Big((1-\cos(t)) E_0 + \cos(t)\Big)\ .
\label{eq energy}
\end{align}
as derived in App.~\ref{app:energy}.
The minimum of $E_\text{DBAC}$ gives the ultimate limit on cooling by the ITE approximation $\ket{\psi_\text{DBAC}}$, paralleling the limits found in conventional AC.
Thus our definition achieves the requirements of processing quantum information both obliviously and coherently: In Eq.~\eqref{eq: GCI state main} there are two single qubit rotations with the Hamiltonian $\h=Z$ and the reflection around the input state $\ket \psi$ as in the Grover algorithm which is key for the energy reduction.

Implementing DBAC requires applying unitary operations given by $U_\psi = e^{it\ket{\psi}\bra{\psi}}$.
When $t = \pi$ then these are the familiar Grover's reflections
$U_\psi = \id - 2\ket \psi\bra\psi$.
However, unlike in Grover's algorithm we assume that we do not know the state $\ket{\psi(0)}$ and yet want to cool it.
Instead, we will use that for any two qubits, in states $\rho$ and $\sigma$, respectively, for $t\ll 1$, and then we have
\begin{align}\label{eq:channel}
	e^{-it\rho} \sigma e^{it\rho} &\approx \text{Tr}_{1}\left[U_{\text{DME}}(\rho\otimes\sigma) U_{\text{DME}}^\dagger\right],
\end{align}
where $U_{\text{DME}} = e^{-it (X\otimes X+Y\otimes Y + Z\otimes Z)/2}$ is the Heisenberg interaction of the two qubits and $X, Y$ and $Z$ are the Pauli matrices. 
This prescription is called density-matrix exponentiation (DME)~\cite{lloyd2014quantum,kjaergaard,kimmel,Wei2023hermpreserving}.
It is dynamic because the physically executed operation $U_{\text{DME}}$ is independent of $\rho$ and $\sigma$ but the quantum information contained in $\rho$ decides what rotation is applied to $\sigma$.
Repeating DME iteratively with multiple copies of $\rho$ allows us to extend the duration $t$.
Thus, the DBAC protocol proceeds as follow: 1) Apply the Hamitonian on the initial state $\ket{\sigma}= e^{-it\h}\ket\psi$ of the data qubit, 2) Perform an appropriate number of DME interactions with instruction qubits, and 3) After $\sigma = \ket\sigma\bra\sigma$ has been reflected around $\ket \psi$ by DME, perform the final rotation $e^{it\h}$ from Eq.~\eqref{eq: GCI state main}.

Bringing in an increasing number of instruction qubits $\rho$ to use in iterative steps to improve cooling performance is  reminiscent of other AC protocols.
The key distinction, however, is that DBAC process the coherent quantum information of pure states without prior knowledge of their content, i.e. it operates obliviously.
This allows us to probe Nernst’s unattainability principle in a new, complementary setting: manipulating quantum coherence rather than mixed-state populations. 
The only entangling operation required is the Heisenberg interaction, in contrast to the more complicated multi-qubit unitaries often used in AC. 
As with entropy-based cooling, a single DBAC step has a finite cooling limit, and perfect implementation would require infinitely many DME steps - another direct analogue of Nernst’s principle. 
Consequently, reaching the ground state demands increasing resources.

One natural way to boost cooling is to apply DBAC recursively. 
After producing $\ket{\psi_{\mathrm{DBAC}}}$ in the first step, we can then define the second step of the recursion as
\begin{align}\label{eq: dbac-2}
  \ket{\psi_\text{DBAC-2}} =e^{it\h} e^{it\ket{\psi_\text{DBAC}}\bra{\psi_\text{DBAC}}}   
    e^{-it\h}\ket{\psi}\ .
\end{align}  
We anticipate that if preparing $\ket{\psi_\text{DBAC-2}}$ requires $M_2$
DME steps, then $M_2$ copies of $\ket{\psi_\text{DBAC}}$ must be prepared in parallel.
If preparing $\ket{\psi_\text{DBAC}}$ requires $M_1$ steps of DME, the recursive protocol required $M = M_1\cdot M_2$ copies of the input state $\ket \psi$.
Similar recursive structures have been considered in other AC protocols as well.

%%%%%%%%%%%%%% Final Experimental results  
\section{Experimental results}

In DBAC adding more qubits yields an improvement, akin to Nernst's principle.
This, together with DBAC involving two-qubit interactions only, makes DBAC appealing for exploring new experimental setups. Here, we demonstrate DBAC on a superconducting quantum lattice implemented in a 3D-integrated cQED architecture in which universal quantum control~\cite{spring2022high, PhysRevLett.133.120802, fasciati2024complementing, bakr2025intrinsic, bakr2025longran}, multiplexed readout~\cite{bakr2024multiplexed, cao2025readoutenhanced}, and algorithms~\cite{Cao_2024, Caoqtr2024} have been demonstrated. The lattice consists of 16 transmon qubits
% by our recent fabrication of co-axial transmon qubits in $4\times 4$ layout 
which allows selecting various subsets of qubits~\cite{alghadeer2025low, alghadeer2025characterization} (See supplementary materials).

Native two-qubit interactions between fixed-frequency transmons form a fundamental advantage of implementing DBAC steps on superconducting qubits. Such native interactions are fast and suffer from less coherent errors compared to fully calibrated two-qubit gates. We implement these entangling operations across the lattice by using the Stark-induced ZZ by level excursions (siZZle) technique \cite{sizzlePRL, PhysRevLettMitchell, bakr2024dynamic} to boost static ZZ coupling (see Fig.~\ref{siZZle_Calibration}a). We observed a large range of siZZle parameters (Fig.~\ref{siZZle_Calibration}b and c) allowing to tune up ZZ interaction strength over various phase parameters to realize the entangling operation $R_{ZZ}(\phi)=e^{-i\phi Z\otimes Z}$.
See supplementary materials for more details on the calibration procedure.

\begin{figure}[H]
  \centering
   \includegraphics[width=0.42\textwidth]{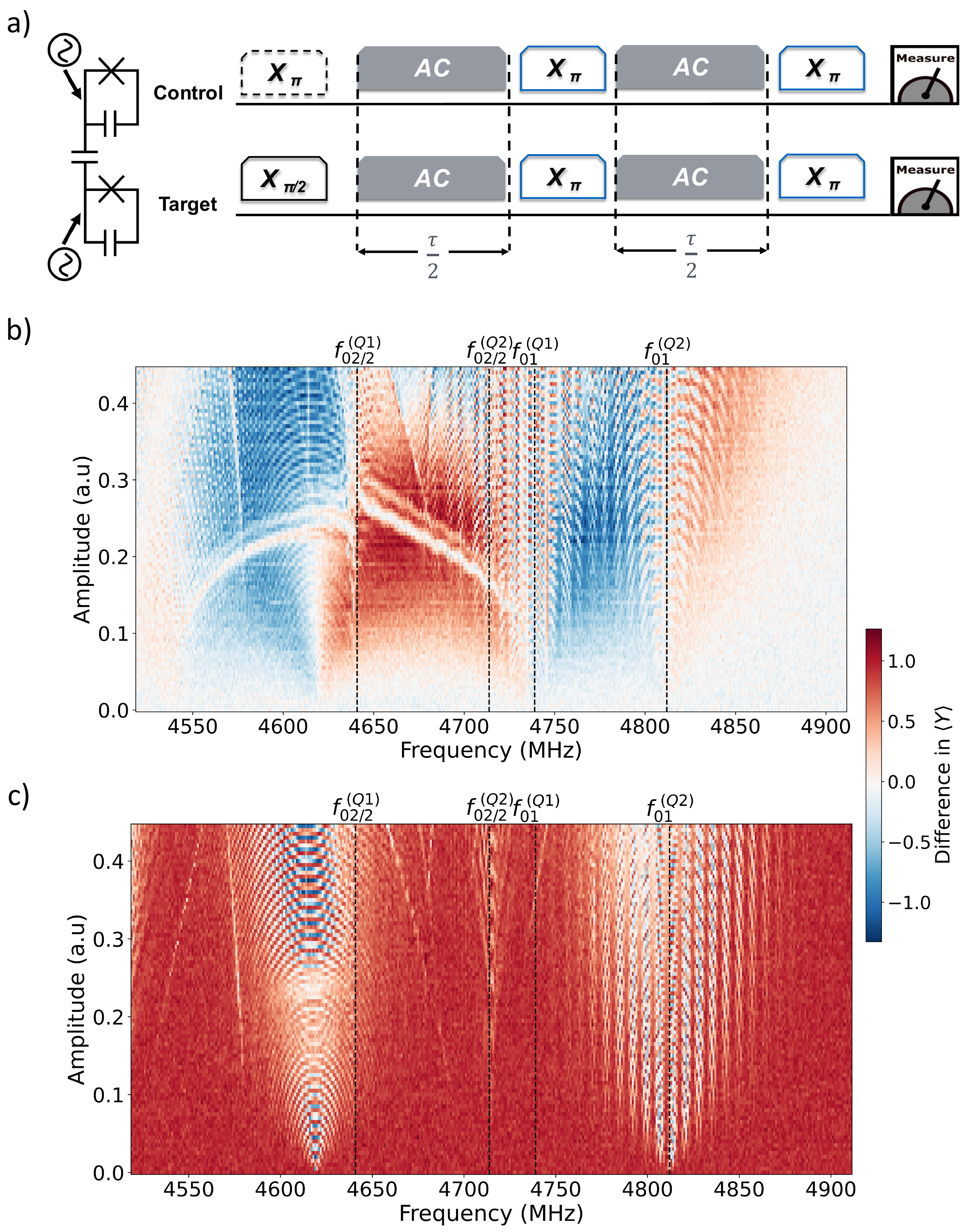}
  \caption{a) Pulse sequence used to calibrate two-qubit $ZZ$ interaction via an echoed Stark-drive protocol. An off-resonant microwave drive (AC Stark tone) is applied to one qubit during free-evolution intervals of total duration $\tau$ (split into $\tau/2$ segments), while interleaved $\pi$ echo pulses on both qubits are used to cancel isolated $IZ$ and $ZI$ rates. This sequence implements the effective unitary $R_{ZZ}(\phi)=e^{-i\phi Z\otimes Z/2}$, where $\phi$ depends on both drives amplitude, frequency, and phase difference. (b) Two-dimensional sweep of Stark-drive frequency and amplitude, showing the response of the target qubit, which is used to tune up the effective $ZZ$ coupling by mapping native drive parameters to the accumulated phase $\phi$. (c) Corresponding measurement on the control qubits, capturing the differential phase accumulation due to the $ZZ$ coupling. These scans are used to select the operating points (Stark-drive drives frequency and amplitude) that implement the effective $R_{ZZ}(\phi)$ operation in the subsequent DBAC experiments.} 
  
  \label{siZZle_Calibration}
\end{figure}

We found that using native $ZZ$ interactions allows to blur the line between quantum computation and quantum simulation: imperfections of experimental assignment of siZZle pulse duration lead merely to the spread of the $ZZ$ interaction duration $R_{ZZ}(\phi+\delta \phi) = e^{-i(\phi+\delta \phi) Z\otimes Z}$, and after we compose it together with single qubit rotations we arrive at continuous deformation of the DME unitary time. This contrasts with unitary synthesis of $U_\text{DME}$ using 3 CZ gates as was demonstrated in Ref.~\cite{kjaergaard} where small miscalibrations in single-qubit or CZ gates can strongly distort the resulting unitary, rather than rescaling of the interaction duration (see Fig.~\ref{DME_Data} for process tomography and supplementary materials for extended results).
Through numerical simulations, we found that the performance of DBAC is robust to changes of the evolution time $t$ in $U_\text{DME}(t)$.

\begin{figure}[H]
    \centering
    \includegraphics[width=.823\linewidth]{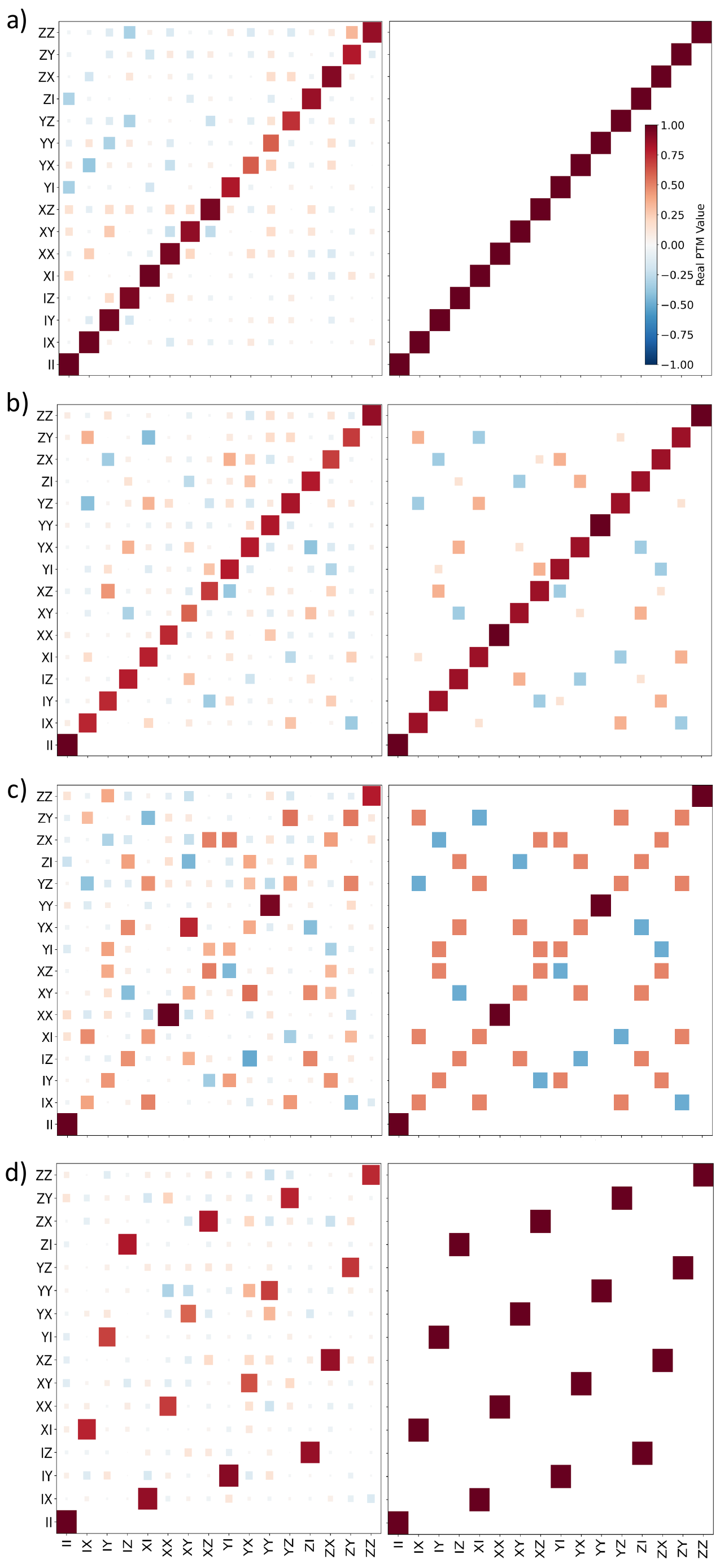}
    \caption{
    Pauli transfer matrix (PTM) characterization of DME compilation obtained via standard quantum process tomography~\cite{nielsen00,Emerson2007}. Panels show experimental PTMs (left) alongside analytic PTMs (right) for different evolution angles: \textit{a)} $\phi=0$, where the $ZZ$ interaction is not involved and deviations from theory arise from readout errors, \textit{b)} $\phi=\pi/8$, \textit{c)} $\phi=\pi/4$, and \textit{d)} $\phi=\pi/2$. The corresponding average process fidelities are \textit{a)} $92.27\%$, \textit{b)} $91.78\%$, \textit{c)} $91.33\%$, and \textit{d)} $90.20\%$. For $\phi=\pi/2$, the DME realizes a SWAP operation (see supplementary materials for more details).
    }
    \label{DME_Data}
\end{figure}

\begin{figure}[H]
  \centering

  % --- (a) ---
  \begin{subfigure}[t]{0.42\textwidth}
    \centering
    \resizebox{\linewidth}{!}{%
      \begin{tikzpicture}
        % \node[fill=white, draw=darkgrayborder, rounded corners=6pt,
        %       line width=1pt, inner sep=3pt]
              {%
          \begin{tikzpicture}
            \node[anchor=south west, inner sep=0] (image) at (0,0)
              {\includegraphics{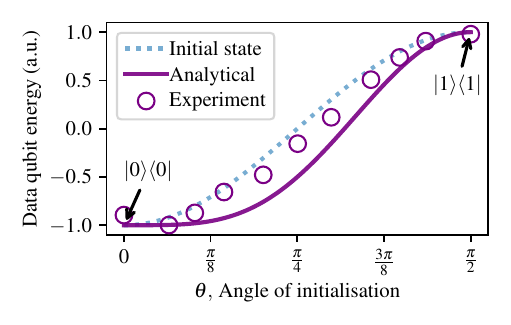}};
            \node[anchor=south, scale=1, yshift=-1.2em, xshift=1em] at (image.north)
              {\mbox{\input{circuits/circuit_A}}};
            \node[anchor=south, yshift=4em] at (image.north)
              {a) 1 DBAC step, 1 DME step};
          \end{tikzpicture}%
        };
      \end{tikzpicture}%
    }
  \end{subfigure}
  \hspace{1.5cm} % <--- extra spacing here
  \hfill
  % --- (b) ---
  \begin{subfigure}[t]{0.45\textwidth}
    \centering
    \resizebox{\linewidth}{!}{%
      \begin{tikzpicture}
        % \node[fill=white, draw=darkgrayborder, rounded corners=6pt,
        %       line width=1pt, inner sep=3pt] 
              {%
          \begin{tikzpicture}
            \node[anchor=south west, inner sep=0] (image) at (0,0)
              {\includegraphics{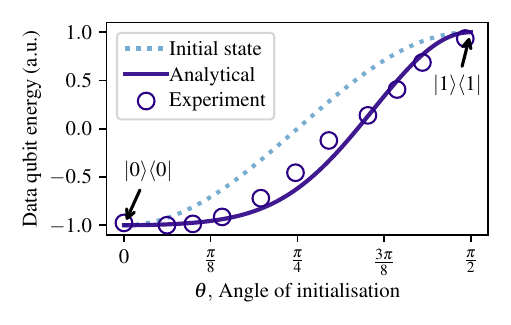}};
            \node[anchor=south, scale=0.9, yshift=-1.2em, xshift=1em] at (image.north)
              {\mbox{\input{circuits/circuit_B}}};
            \node[anchor=south, yshift=5em] at (image.north)
              {b) 1 DBAC step, 2 DME steps};
          \end{tikzpicture}%
        };
      \end{tikzpicture}%
    }
  \end{subfigure}
  \hspace{1.5cm} % <--- extra spacing here
  \hfill
  % --- (c) ---
  \begin{subfigure}[t]{0.45\textwidth}
    \centering
    \resizebox{\linewidth}{!}{%
      \begin{tikzpicture}
        % \node[fill=white, draw=darkgrayborder, rounded corners=6pt,
        %       line width=1pt, inner sep=3pt] 
              {%
          \begin{tikzpicture}
            \node[anchor=south west, inner sep=0] (image) at (0,0)
              {\includegraphics{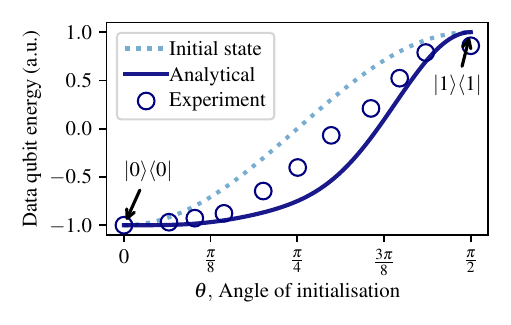}};
            \node[anchor=south, scale=0.77, yshift=-1.2em, xshift=1em] at (image.north)
              {\mbox{\input{circuits/circuit_C}}};
            \node[anchor=south, yshift=6em] at (image.north)
              {c) 2 DBAC steps, 1 DME step};
          \end{tikzpicture}%
        };
      \end{tikzpicture}%
    }
  \end{subfigure}

  \caption{DBAC as a function of initialization angle $\theta$ (i.e., apply DBAC to $\ket{\psi}=R_X(\theta)\ket 0$) for \textit{a)} the minimal case with $n=2$ qubits, \textit{b)} with $n=3$ qubits, and \textit{c)} with $n=4$ qubits. Experimental data show that distributing DME can improve cooling, with a flattening of the energy curve. Analytic simulations predict further cooling with more DBAC steps, but decoherence places a break-even limit on our device.}
  \label{fig:dbac-3-circuits}
\end{figure}

In each step of DBAC, we have the freedom to optimize the step duration, which can be selected to bring a specific portion of the Bloch sphere within a target fidelity to the ground state.
Based on the optimization, we found $t=\pi/4$ for the first and second DBAC steps yields cooling across a wide range of initial states (see Fig.~\ref{fig:dbac-3-circuits}a) for $n=2$ qubits).
In Fig.~\ref{fig:dbac-3-circuits}b), we show that the amount of cooling is larger when DME is implemented with 2 instruction qubits instead of 1, and indeed our experimental data for this case shows the best DBAC performance in agreement with the analytic simulations.
Fig.~\ref{fig:dbac-3-circuits}c) shows that performing 2 DBAC steps lead to even larger cooling, but despite being implemented by 
more qubits ($n=4$), we found that the data is at the verge of break-even due to increased circuit depth and our relatively long durations of the siZZle pulses.

\section{Discussion and Conclusion}
Nernst’s unattainability principle, formulated in 1912, states that “it is impossible for any procedure to lead to the isotherm $T=0$ in a finite number of steps”~\cite{nernst1912}. %Guggenheim1933
In the quantum regime, an extreme scenario is to consider $n$ qubits all in the same pure state $\ket{\psi}$ and ask: 
How to transform at least one qubit into the ground state $\ket{0}$?
A natural approach would be to perform tomography by making measurements on $n-1$ qubits, where to improve the precision of the rotation to ground state of the $n$-th qubit requires increasing the total number of qubit $n$. However, the central limit theorem dictates that the statistical error of  an estimator such as $\langle X\rangle$ decreases only as $1/\sqrt{n}$ and even adaptive online learning methods need large sample sizes~\cite{MarcoRegret}. 

This increase in number of copies for improving precision illustrates in a quantum setting the same unattainability principle: approaching the exact ground state demands asymptotically increasing resources. In contrast, we presented a coherent approach reminiscent of algorithmic cooling, but instead of reducing entropy from the target qubit, DBAC removes quantum coherence.
Our approach is dynamic~\cite{kjaergaard} and uses density matrix exponentiation to encode quantum information stored in copies of the qubits into a quantum operation.
This operation uses the pure density matrix $\ket\psi\bra\psi$ of the qubits at a specific moment as a temporary Hamiltonian to govern the target qubit evolution.
While dynamic quantum algorithms can perform universal quantum computation~\cite{kimmel}, DBAC is, to our knowledge, the first protocol where the dynamic aspect naturally aligns with the objective and has a clear physics utility.
The dynamic density-matrix exponentiation is interlaced with echo operations involving the Hamiltonian whose ground state is the target state. 
The algorithm's performance is grounded in Riemannian gradient flows and imaginary-time evolution~\cite{dbqite}, but also we proved analytical formulas on the ultimate cooling limit of one such echoing step.
By forming a recursion~\cite{gluza2024double,qdp,dbqite,dbgrover}, DBAC can converge towards the perfect ground state at the cost of increasing number of qubits $n$,
a trade-off reminiscent of Nernst's principle and similar to limits in other AC protocols.

We implemented DBAC experimentally on a superconducting lattice of transmon qubits and investigated its recursive application to cooling operations.
This demonstrates, for the first time, the utility of employing quantum dynamic programming~\cite{qdp}.
We found that DBAC stands out as particularly accessible for experiments which can be shown by comparison with the paradigmatic HBAC protocol (see App.~\ref{app:E}).
To operate in our scenario, HBAC requires a dephasing operation before proceeding to remove the resultant increase in system entropy. 
In contrast, DBAC directly removes quantum coherence.
Even in the minimal $n=3$ case, the entropy removal process by HBAC involves 3 Toffoli gates, which translates generally to about two dozen CZ gates.
In our experiment, running a single step of DBAC achieved comparable process fidelity to a single CZ gate, effectively rendering HBAC impractical. 
This advantage stems from our compilation strategy, which was rooted in utilizing native ZZ interactions, allowing to reduce the physical runtime of the protocol.
Unlike CZ-based compilation, precise calibrations are not as essential since DBAC performance depends continuously on durations of ZZ interactions.

The primary application of HBAC has been NMR and NV-center systems, where it is used for removing mixedness of qubits as represented by pseudo-pure states $\rho = \tfrac p2 \id +(1-p)\ket\psi\bra\psi$~\cite{pseudopure1,pseudopure2}.
In this setting DBAC would yield states $\rho' \approx \tfrac p2 \id +(1-p)\ket0\bra0$ which for generic $\ket \psi$ have decreased polarization compared to dephasing of $\rho$  and thus, after coherence removal by DBAC, the cooling limit of HBAC is reduced~\cite{CorrelationEnahncedAC,OverviewAACPhysRevA.110.022215,AchievablePolarizationAC}.
It is not a disadvantage that DBAC cannot increase purity of initially mixed states, since this can be performed by HBAC or  probabilistic protocols requiring entangled measurements~\cite{Cirac1999_purification}; instead, DBAC enables the complementary capability to reduce quantum coherence.

DBAC's dynamic character is key: Copies of input states act as quantum instructions which program the operation of the algorithm~\cite{kjaergaard}.
DBAC is a minimal example of applying quantum dynamic programming~\cite{qdp} to a double-bracket quantum algorithm~\cite{gluza2024double}.
For multi-qubit states $\ket \psi$, double-bracket quantum imaginary-time evolution has convergence guarantees rooted in Riemannian steepest-decent flows and thus AC of quantum coherence can be generalized to larger quantum systems.
By this, DBAC extends the utility of dynamic quantum algorithms and paves the way for using them to tackle foundational objectives of quantum thermodynamics.

% \clearpage
\onecolumngrid

\section*{Acknowledgments}
M.B. acknowledges support from EPSRC QT Fellowship grant EP/W027992/1, and EP/Z53318X/1.
P.L. acknowledges support from EP/N015118/1, and
EP/T001062/1. KG, MG and NN are supported by the start-up grant of the Nanyang Assistant Professorship, and the Ministry of Education, Singapore under its Academic Research Fund Tier 1 (RT1/23). 
Additionally, KG and MG acknowledge funding by the Presidential Postdoctoral Fellowship of the Nanyang Technological University.
%S.C. acknowledges support from Schmidt Science (not an active grant). 
We would like to acknowledge the Superfab Nanofabrication facility at Royal Holloway, University of London, and Optoelectronics Research Centre at University of Southampton where part of device fabrication was performed. 

\section*{Data Availability}
Data and experimental code supporting the findings of this work are publicly available at
\href{https://github.com/MoGhadeer/dbac-algorithmic-cooling/tree/main}{github.com/MoGhadeer/dbac-algorithmic-cooling/tree/main}.

% \newpage
% \clearpage

% \clearpage
% \newpage
% \newpage
\onecolumngrid
\appendix
% \tableofcontents
\def\1{I}

% \textbf{Supplementary Materials}
\vspace{3cm}
%\vspace*{1cm} % adds vertical space (adjust 2cm as needed)
{\centering \textbf{Supplementary Materials for ``Double-Bracket Algorithmic Cooling"} \par}
%%%%%%%%%%%%%%%%%%%%%%%%%%%%%%%%%%%%%%%%%%%%%%

\section{Compilation}
\subsection{DME implementation using native ZZ interactions}
In this section, we provide necessary details on the density matrix exponentiation (DME) algorithm, which is the main ingredient that executes DBAC. Given the instruction register $\rho$ and the data register $\sigma$ of the same dimension, the DME channel transforms the data registers into an approximation of the unitary $e^{-it\rho}$.
The circuit used to execute the DME circuit remains the same (see Fig.~\ref{fig:DME}), however the operation performed changes depending on the input register. 
This highlight the key feature of a dynamic quantum algorithm, which is that it is oblivious to (does not require prior knowledge of) input states. 
The DME channel $\hat{E}_t^{(\rho)}(\sigma)$ acts on the data qubit $\sigma$ according to the following relation:
\begin{align}\label{eq:channel}
	\hat{E}_t^{(\rho)}(\sigma) &= \text{Tr}_{1}\left[e^{-i t \swap}(\rho\otimes\sigma) e^{i t \swap}\right] = \sigma - it[\rho,\sigma] +O(t^2)
        = e^{-it\rho} \sigma e^{it\rho} +O(t^2)
\end{align}
where $t\in \mathbb R$ is the DME channel duration, and $\swap$ is simply defined as $\swap\ket{v}\otimes\ket{w} = \ket{w}\otimes\ket{v}$.
\begin{figure}[H]
    \centering
        \scalebox{1}{
        \Qcircuit @C=1.0em @R=1.0em @!R { \\
        & \lstick{{q}_{0} :  } & \qswap & \qw & \qw & \qswap & \qw & \qw & \dots & \hspace{1em} & \qswap & \qw & \qw \\
        & \lstick{{q}_{1} :  } & \qswap \qwx[-1] & \meter \\  
        & \lstick{{q}_{2} :  } & \qw & \qw &  \qw & \qswap \qwx[-2] & \meter \\
        & \vdots \hspace{3em}  & \qw & \qw & \qw & \qw & \qw & \qw & \dots & \hspace{1em} \\
        & \lstick{{q}_{M} :  } & \qw & \qw & \qw & \qw & \qw & \qw &  \dots & \hspace{1em} & \qswap \qwx[-4] & \meter
        \gategroup{2}{3}{3}{4}{.8em}{--} % Adds box around first swap
        \gategroup{2}{6}{4}{7}{.8em}{--} % Adds box around second swap
        \gategroup{2}{11}{6}{12}{.8em}{--} % Adds box around Mth swap
        }}
    \caption{Circuit performing $M$ iterations of density matrix exponentiation (DME). The data qubit is denoted $q_0 (= \sigma)$, while $q_1, \dots, q_M$ represent $M$ copies of the instruction qubit $\rho$. Each iteration of DME, illustrated within the dashed box, consists of a $\delta\text{-swap}$ operation followed by a trace-out procedure. In this diagram, the trace-out is represented by a mid-circuit measurement, though the qubit can be disregarded. The circuit approximates the unitary operation $e^{-it\rho}$ acting on the data qubit $q_0$.}
    \label{fig:DME}
\end{figure}
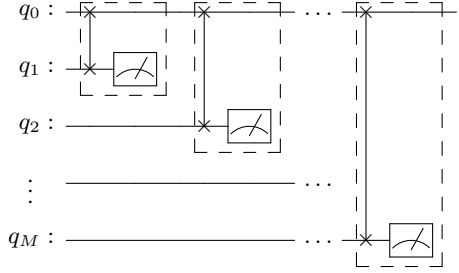

Setting $t\mapsto t/M$, we find that repeating this operation $M$ times with new copy of the input $\rho$ yields 
\begin{align}
    \left(\hat{E}_{(t/M)}^{(\rho)}(\sigma)\right)^{M} = e^{-it\rho} \sigma e^{it\rho} + O(t^2/M),
    \label{eq:error}
\end{align}
i.e. the operation converges to a partial reflector. This requires $M$ iterations of the DME channel and thus consumes $M$ copies of the state $\rho$. 
It has a similar effect as Trotterizing M times. 
We can see that the final effect of the Trotterized channel is applying the unitary $u_\rho$ to accuracy within $O(t^2/M)$ (see Fig. ~\ref{fig:DME-norm}).

\begin{figure}[H]
    \centering
    \includegraphics[width=0.6\linewidth]{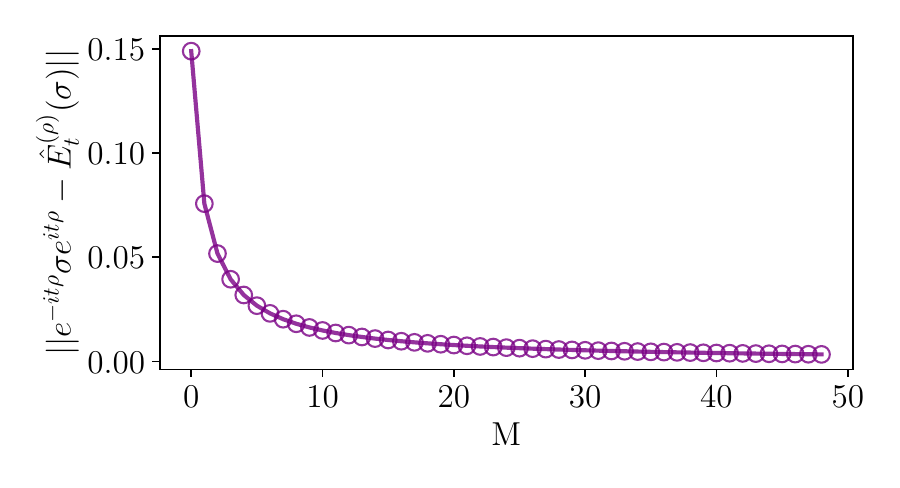}\vspace{-0.4cm}
    \caption{Performance of DME channel $\hat{E}_t^{(\rho)}(\sigma)$ after M iterations to approximate the evolution given by $e^{-it\rho} \sigma e^{it\rho}$.}
    \label{fig:DME-norm}
\end{figure} 

In this work, we employ the fractional $\swap$ gate with arbitrary duration $\delta\swap = e^{-it\swap}$, using an alternative decomposition that utilizes hardware-native interactions, rather than the standard decomposition into three CNOT gates and single-qubit unitaries ~\cite{VatanWilliams,Vidal_2004}. Specifically, the Heisenberg interaction can be leveraged:
\begin{align}
    \hat{H}_{i,j} &= X_i X_j + Y_i Y_j + Z_i Z_j,
    \label{eq:Heisenberg}
\end{align}
which allows the SWAP operator acting on qubits $i$ and $j$ to be expressed as $\swap_{i,j} = \frac{1}{2}(I + \hat{H}_{i,j})$. Thus, a time-evolved $\delta \swap$ gate can be realized via the unitary evolution:
\begin{align}
     \delta \swap_{i,j}(t) &= e^{-it \frac{1}{2}(I + X_i X_j + Y_i Y_j + Z_i Z_j)}.
\end{align}

The Heisenberg interaction $\hat{H}$ (Eq.~\eqref{eq:Heisenberg}) can be used to implement a two-qubit SWAP operation of duration $\delta t$ using a native $ZZ$ interaction (see Fig.~\ref{fig:heisenberg_circuit}). 
This approach is feasible due to the transformations $R_X\left(\pi/2\right)ZR_X\left(-\pi/2\right) = Y$ and $R_Y\left(\pi/2\right)ZR_Y\left(-\pi/2\right) = X$.
Applying $R_Y\left(\pi/2\right)$ rotations to both qubits maps a $Z \otimes Z$ interaction to $X \otimes X$, while $R_X\left(\pi/2\right)$ rotations yield a $Y \otimes Y$ interaction.
We chose $R_X$ and $R_Y$ gates to realize the Heisenberg interaction, as these are native single-qubit gates in our hardware.
It is also possible to use typical universal set of single-qubit gates such as phase gate and Hadamard gate to realize the Heisenberg interaction (see Supplementary Materials).

\subsubsection{DME function used in DBAC}

Using the fact that SWAP gate is unitary, one has
\begin{equation}
    e^{i \, \text{SWAP} \, \delta} = \cos(\delta)\, \mathbb{I} + i \sin(\delta)\, \text{SWAP}.
\end{equation}
The corresponding DME channel acting on a data qubit state $\sigma$ with an instruction qubit state $\rho$ is
\begin{align}
    \mathcal{E}_{\text{DME}}(\sigma) 
    &= \mathrm{Tr}_2 \!\left( e^{-i \, \text{SWAP} \, \delta} \, (\sigma \otimes \rho) \, e^{i \, \text{SWAP} \, \delta} \right) \\
    &= \cos^2(\delta)\,\sigma + i \cos(\delta)\sin(\delta)\,[\rho,\sigma] + \sin^2(\delta)\,\rho.
\end{align}
For clarity, we explicitly note that throughout this section we use the notation
\begin{equation}
    U_{\rm DME}(\delta) \equiv e^{i \, \delta \, \text{SWAP}},
\end{equation}
so that subsequent expressions are written in terms of $U_{\rm DME}$.
\\
\subsubsection{DME Circuit}
For convenience, we now denote
\begin{equation}
    U_{\rm DME}(\phi) \equiv e^{i \phi \, \text{SWAP}},
\end{equation}
so that the circuit representations below are written in terms of $U_{\rm DME}$ rather than the exponential form. 

\begin{center}
    \begin{minipage}{\textwidth}
    \centering
    $\vcenter{\hbox{\scalebox{0.9}{\Qcircuit @C=1em @R=0.8em @!R { \\
         & \multigate{1}{\mathrm{U_{DME}({\phi})}} & \qw \\
         & \ghost{\mathrm{U_{DME}({\phi})}} & \qw \\
\\ }}}}
    =
    \vcenter{\hbox{\Qcircuit @C=1.0em @R=0.2em @!R { \\
    & \ctrl{1} & \dstick{\hspace{2.0em}\mathrm{R_{ZZ}}\,(\mathrm{\phi})} \qw & \qw & \qw \barrier[0em]{1} & \qw & \gate{\mathrm{S^\dagger}} & \gate{\mathrm{H}} & \ctrl{1} & \dstick{\hspace{2.0em}\mathrm{R_{ZZ}}\,(\mathrm{\phi})} \qw & \qw & \qw & \gate{\mathrm{H}} & \gate{\mathrm{S}} \barrier[0em]{1} & \qw & \gate{\mathrm{H}} & \ctrl{1} & \dstick{\hspace{2.0em}\mathrm{R_{ZZ}}\,(\mathrm{\phi})} \qw & \qw & \qw & \gate{\mathrm{H}} \barrier[0em]{1} & \qw & \qw \\
    & \control \qw & \qw & \qw & \qw & \qw & \gate{\mathrm{S^\dagger}} & \gate{\mathrm{H}} & \control \qw & \qw & \qw & \qw & \gate{\mathrm{H}} & \gate{\mathrm{S}} & \qw & \gate{\mathrm{H}} & \control \qw & \qw & \qw & \qw & \gate{\mathrm{H}} & \qw & \qw \\ \\ \\
} }}$
\end{minipage}
\end{center}
\vspace{-1.5em}

This approach is feasible due to the transformations $HZH = X$ and $(S^\dagger H) Z (H S) = Y$, where $H$ denotes the Hadamard gate (distinct from the Heisenberg interaction $\hat{H}$). By applying Hadamard gates on both qubits, a $ZZ$ ($Z \otimes Z$) interaction can be transformed into an $XX$ ($X \otimes X$) interaction:
\begin{align}
    (H \otimes H) e^{-itZZ} (H \otimes H) &= e^{-itXX}.
\end{align}
Similarly, applying the $S^\dagger$ gate followed by a Hadamard gate enables the realization of a $YY$ ($Y \otimes Y$) interaction:
\begin{align}
    (S^\dagger H \otimes S^\dagger H) e^{-itZZ}  [(H S) \otimes (H S)] &= e^{-itYY}.
\end{align}

This is equivalent to

\begin{figure}[h]
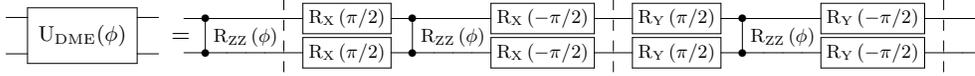

    \centering
    \scalebox{1}{  % 0.53 for half width
    \begin{minipage}{\textwidth}
        \centering
        $\vcenter{\hbox{\hspace*{-5em}\input{circuits/dme}}}
        =
        \vcenter{\hbox{\hspace*{-3em} \input{circuits/dme_compile_experiment}}}$
    \end{minipage}
    }
    \caption{The Heisenberg interaction is employed to simulate a $\delta$-SWAP gate, which is a partial SWAP operation that evolves for a phase $\phi = \omega t$, where $\omega$ is the rate of ZZ interaction, and $t$ is the evolution duration. This process also represents a single Trotterization step ($M = 1$) of density matrix exponentiation (DME), which is the exponentiation $e^{-i\phi\rho}$ of density matrix $\rho$ with phase $\phi$.}
    \label{fig:heisenberg_circuit}
\end{figure}

The Pauli transfer Matrix for the Heisenberg compilation of DME (see Fig.~\ref{fig:heisenberg_circuit} for the detailed circuit) is given in Fig.~\ref{DME_Data}.

%%%%%%%%%%%%%%%%%%%%%%%%%%%%%%%%%%%%%%%%%%%%%%%%%%%%%%%%%%%%%%%%%%%%%%%%%%%%%%%%%%%%%%%%%%%%%

\subsection{Detailed circuit compilations}

\subsubsection{Full DBAC demonstration circuits}
Table~\ref{tab:k_m_values_M1} provides detailed circuit configurations corresponding to the qubit layout in Figure~\ref{fig:layout}. These circuits are optimized to balance computational feasibility with hardware limitations while maintaining accurate Trotterized evolution. Circuit B implements an $M = 2$ circuit to demonstrate a novel realization of DME. In contrast to the approach in Ref.~\cite{kjaergaard}, which applies DME with instruction recycling on a two-qubit system, our implementation extends to three qubits without recycling instructions. Circuit C implements a $k = 2$ circuit to highlight the recursive nature of the DBAC algorithm.

% \begin{table}[H]
% \centering
% \begin{tabularx}{\textwidth}{c|c|>{\centering\arraybackslash}X}
%     $k$ & $M$ & Circuit \\
%     \hline
%     1 & 1 & 
%     \input{circuits/circuit_A} \\
%     \\
%     && \textbf{Circuit A}: Reset qubit $q_7$, initialized using $R_X(\theta)$, with DME angle $\phi=\pi/4$, uses $n=2$ qubits. 
%     % See \ref{fig:qpt_k1M1} \nn{[broken link]} for the process tomography of this circuit. 
%     \\ 
%     \hline
%     1 & 2 & 
%     \input{circuits/circuit_B} \\
%     \\ 
%     && \textbf{Circuit B}: $\phi=\pi/4$, uses $n=3$ qubits. This shows 2 steps of DME Trotterization. \\
%     \hline    
%     2 & 1 &  
%     \input{circuits/circuit_C} \\
%     \\
%     && \textbf{Circuit C}: $\phi=\pi/4$, uses $n=4$ qubits. This shows 2 steps of the recursive algorithm, separated by the barrier (dotted line). 
%     The first step $k=1$ is identical to Circuit A.
% \end{tabularx}

\begin{table}[H]
\centering
\begin{tabularx}{\textwidth}{c|c|>{\centering\arraybackslash}X}
\hline
$k$ & $M$ & Circuit \\
\hline

1 & 1 &
\begin{minipage}[c]{\linewidth}
\centering
\scalebox{1.0}{
    \Qcircuit @C=1.0em @R=0.2em @!R { \\
    	 	\lstick{q_1 :  } & \gate{\mathrm{R_X}\,(\mathrm{\theta})}\barrier[0em]{1} & \qw & \qw & \multigate{1}{\mathrm{U_{DME}({\pi/4})}} & \qw \\
    	 	\lstick{q_2 :  } & \gate{\mathrm{R_X}\,(\mathrm{\theta})} & \qw & \gate{\mathrm{R_Z}\,(\mathrm{\pi/4})} & \ghost{\mathrm{U_{DME}({\pi/4})}} & \qw \\
    \\ }}

\vspace{0.5em}

\textbf{Circuit A}: Reset qubit $q_7$, initialized using $R_X(\theta)$, with DME angle $\phi=\pi/4$, uses $n=2$ qubits.
\end{minipage} \\
\hline

1 & 2 &
\begin{minipage}[c]{\linewidth}
\centering
\scalebox{1.0}{
    \Qcircuit @C=1.0em @R=0.2em @!R { \\
    	 	\lstick{q_1 :  } & \gate{\mathrm{R_X}\,(\mathrm{\theta})}\barrier[0em]{2} & \qw & \gate{\mathrm{R_Z}\,(\mathrm{\pi/4})} & \multigate{1}{\mathrm{U_{DME}({\pi/8})}} & \qw & \qw \\
    	 	\lstick{q_2 :  } & \gate{\mathrm{R_X}\,(\mathrm{\theta})} & \qw & \qw & \ghost{\mathrm{U_{DME}({\pi/8})}} & \multigate{1}{\mathrm{U_{DME}({\pi/8})}} & \qw \\
    	 	\lstick{q_3 :  }  & \gate{\mathrm{R_X}\,(\mathrm{\theta})} & \qw & \gate{\mathrm{R_Z}\,(\mathrm{\pi/4})} & \qw & \ghost{\mathrm{U_{DME}({\pi/8})}} & \qw \\
    \\ }}

\vspace{0.5em}

\textbf{Circuit B}: $\phi=\pi/4$, uses $n=3$ qubits. This shows 2 steps of DME Trotterization.
\end{minipage} \\
\hline

2 & 1 &
\begin{minipage}[c]{\linewidth}
\centering
\scalebox{1.0}{
    \Qcircuit @C=1.0em @R=0.2em @!R { \\
    	 	\lstick{q_1 :  } & \gate{\mathrm{R_X}\,(\mathrm{\theta})}\barrier[0em]{3} & \qw & \gate{\mathrm{R_Z}\,(\mathrm{\pi/4})} & \multigate{1}{\mathrm{U_{DME}({\pi/4})}} \barrier[0em]{3} & \qw & \qw & \qw & \qw \\
    	 	\lstick{q_2 :  } & \gate{\mathrm{R_X}\,(\mathrm{\theta})} & \qw & \qw & \ghost{\mathrm{U_{DME}({\pi/4})}} & \qw & \qw & \multigate{1}{\mathrm{U_{DME}({\pi/4})}} & \qw \\
    	 	\lstick{q_3 :  } & \gate{\mathrm{R_X}\,(\mathrm{\theta})} & \qw & \qw & \multigate{1}{\mathrm{U_{DME}({\pi/4})}} & \qw & \gate{\mathrm{R_Z}\,(\mathrm{\pi/4})} & \ghost{\mathrm{U_{DME}({\pi/4})}} & \qw \\
    	 	\lstick{q_4 :  }  & \gate{\mathrm{R_X}\,(\mathrm{\theta})} & \qw & \gate{\mathrm{R_Z}\,(\mathrm{\pi/4})} & \ghost{\mathrm{U_{DME}({\pi/4})}} & \qw & \qw & \qw & \qw \\
    \\ }}

\vspace{0.5em}

\textbf{Circuit C}: $\phi=\pi/4$, uses $n=4$ qubits. This shows 2 steps of the recursive algorithm, separated by the barrier (dotted line). The first step $k=1$ is identical to Circuit A.
\end{minipage} \\
\hline
\end{tabularx}
\caption{Circuit layouts for DBAC with varying recursion steps $k$ and DME Trotter steps $M$. Each circuit corresponds to a different $(k, M)$ pair. All parameters given are angle (in radians). Qubits are arranged to display a nearest-neighbor structure for ease of illustration. See Fig.~\ref{fig:layout} for their actual placement on the square lattice architecture, where two-qubit gates are also nearest-neighbor.}
\label{tab:k_m_values_M1}
\end{table}

% \newpage

\subsubsection{Single- and two-qubit gates and operations}
In this sub-subsection, the detailed definitions of the gate used in the experiment are given.

\begin{center}
\begin{tabular}{c|c}
% \hline
% \textbf{Single-Qubit Gates} & \textbf{Two-Qubit Gates} \\
Single-Qubit Gates & Two-Qubit Operations \\ \\
\hline
\(
\begin{aligned}
I &= \begin{bmatrix} 1 & 0 \\ 0 & 1 \end{bmatrix}, \\ \\
X &= \begin{bmatrix} 0 & 1 \\ 1 & 0 \end{bmatrix}, \\ \\
Y &= \begin{bmatrix} 0 & -i \\ i & 0 \end{bmatrix}, \\ \\
Z &= \begin{bmatrix} 1 & 0 \\ 0 & -1 \end{bmatrix}, \\ \\
H &= \frac{1}{\sqrt{2}} \begin{bmatrix} 1 & 1 \\ 1 & -1 \end{bmatrix}, \\ \\
S &= \begin{bmatrix} 1 & 0 \\ 0 & i \end{bmatrix}, \\ \\
S^\dagger &= \begin{bmatrix} 1 & 0 \\ 0 & -i \end{bmatrix}, \\ \\
R_Z(\phi) &= e^{-i \phi Z}
\end{aligned}
\)
&
\(
\begin{aligned}
R_{ZZ}(\phi) &= \operatorname{diag}\left(
e^{-i\phi/2},\ 
e^{i\phi/2},\ 
e^{i\phi/2},\ 
e^{-i\phi/2}
\right), \\ \\
R_{YY}(\phi) &= (R_X(-\pi/2) \otimes R_X(-\pi/2))\, R_{ZZ}(\phi)\, (R_X(\pi/2) \otimes R_X(\pi/2)), \\ \\
R_{XX}(\phi) &= (R_Y(-\pi/2) \otimes R_Y(-\pi/2))\, R_{ZZ}(\phi)\, (R_Y(\pi/2) \otimes R_Y(\pi/2)), \\ \\
U_\text{DME}(\phi) &= R_{ZZ}(\phi) \cdot R_{YY}(\phi) \cdot R_{XX}(\phi).
\end{aligned}
\)
% \\
% \hline
\end{tabular}

\end{center}

\subsubsection{Other Circuits}

\begin{tabular}{c|c}
% \hline
Name & Circuit \\ \\
\hline
ZZ Circuit &
\scalebox{1.0}{ 
\Qcircuit @C=1.0em @R=0.5em @!R {
  \nghost{} & \gate{\mathrm{R_X}(\pi)} & \qw & \ctrl{1} & \dstick{\hspace{2.0em}\mathrm{ZZ}(\pi/8)} \qw & \qw & \qw & \qw &
  \gate{\mathrm{R_X}(\pi)} & \ctrl{1} & \dstick{\hspace{2.0em}\mathrm{ZZ}(\pi/8)} \qw & \qw & \qw & \qw &
  \gate{\mathrm{R_X}(\pi)} & \qw  \gategroup{1}{4}{2}{15}{1.8em}{--} \\
  \nghost{} & \gate{\mathrm{R_X}(\pi/2)} & \qw & \control \qw & \qw & \qw & \qw & \qw &
  \gate{\mathrm{R_X}(\pi)} & \control \qw & \qw & \qw & \qw & \qw &
  \gate{\mathrm{R_X}(\pi)} & \qw  \\ \\
}} \\ 
\hline
CZ &
\scalebox{1.0}{
\Qcircuit @C=1.0em @R=0.5em @!R {
  \nghost{q_0 :} & \ctrl{1} & \dstick{\hspace{2.0em}\mathrm{ZZ}(\pi/2)} \qw & \qw & \qw & \gate{\mathrm{R_Z}(\pi)} & \gate{\mathrm{S}^\dagger=R_Z(\pi/2)} & \qw  \\
  \nghost{q_1 :} & \control \qw & \qw & \qw & \qw & \gate{\mathrm{R_Z}(\pi)} & \gate{\mathrm{S}^\dagger=R_Z(\pi/2)} & \qw  \\ \\
}} \\ 
\hline
SWAP using 3 CNOTs &
\scalebox{1.0}{
\Qcircuit @C=1.0em @R=0.5em @!R {
  \nghost{q_0 :} & \ctrl{1} & \targ & \ctrl{1} & \qw  \\
  \nghost{q_1 :} & \targ & \ctrl{-1} & \targ & \qw  \\ \\
}} \\ 
\hline
CNOT from CZ &
\scalebox{1.0}{
\Qcircuit @C=1.0em @R=0.5em @!R {
  \nghost{q_0 :} & \qw & \control \qw & \qw & \qw & \qw \\
  \nghost{q_1 :} & \gate{H} & \ctrl{-1} & \gate{H} & \qw & \qw \\ \\
}} \\ 
\hline
SWAP compilation &
\scalebox{0.8}{
\Qcircuit @C=1.0em @R=0.5em @!R {
  \nghost{q_0 :} & \qw & \ctrl{1} & \dstick{\hspace{2.0em}\mathrm{ZZ}(\pi/2)} \qw & \qw & \qw & \gate{\mathrm{R_Z}(\pi)} & \gate{S} & \gate{H} & \ctrl{1} & \dstick{\hspace{2.0em}\mathrm{ZZ}(\pi/2)} \qw & \qw & \qw & \gate{\mathrm{R_Z}(\pi)} & \gate{S} & \gate{H} & \ctrl{1} & \dstick{\hspace{2.0em}\mathrm{ZZ}(\pi/2)} \qw & \qw & \qw & \gate{\mathrm{R_Z}(\pi)} & \gate{S} & \qw & \qw & \qw \\
  \nghost{q_1 :} & \gate{H} & \control \qw & \qw & \qw & \qw & \gate{\mathrm{R_Z}(\pi)} & \gate{S} & \gate{H} & \control \qw & \qw & \qw & \qw & \gate{\mathrm{R_Z}(\pi)} & \gate{S} & \gate{H} & \control \qw & \qw & \qw & \qw & \gate{\mathrm{R_Z}(\pi)} & \gate{S} & \gate{H} & \qw & \qw \\ \\
}} % \\
% \hline
\end{tabular}

\section{Theory}
\label{app:energy}
\subsection{Cooling limits of DBAC}
\subsubsection{ITE solution for a single qubit}
Consider a qubit in initial state $\ket {\psi_0} = c_0 \ket 0 + c_1\ket 1$. For the case of a single qubit, we can choose 
\begin{align}
    H = -Z = \ket 1 \bra 1\ - \ket 0 \bra 0\ ,
\end{align}
as the Hamiltonian for DBAC, the ground state $\ket 0$ is an eigenstate of $Z$ with the lowest eigenvalue $E_g = -1$. To avoid confusion, we reserve $E_g$ for the ground-state energy and use $E_0$ to denote the initial energy of the state $\ket{\psi_0}$.  
\begin{align}
     \ket {\psi(\tau)} &= \frac{c_0 e^{\tau} \ket 0 + c_1 e^{-\tau} \ket 1}{(|c_0|^2 e^{2t} + |c_1|^2 e^{-2t})^{1/2}}\ ,
\end{align}
by using $    e^{-H\tau} = e^{Z\tau} = e^{\tau} \ket 0 \bra 0 + e^{-\tau} \ket 1 \bra 1$.
We notice that the fidelity of the state $\ket{\psi_0}$ with the ground state $\ket 0$ is
\begin{align}
    F_0 = |\langle{0 | \psi_0}\rangle|^2 = |c_0|^2\ ,
\end{align}
while the initial energy is
\begin{align}
    E_0 = |c_1|^2 - |c_0|^2 = 1=2F_0\ .
\end{align}
When we reset $\ket{\psi(\tau)} \to \ket 0$, then $E(\tau) \to 0$, which implies $F(\tau)\to1$.
We are interested in the energy of the state $\ket{\psi(\tau)}$, which is given by
\begin{align}
E(\tau)
= \langle\psi(\tau)|H|\psi(\tau)\rangle
= -\,\frac{|c_0|^{2} e^{2\tau} - |c_1|^{2} e^{-2\tau}}{|c_0|^{2} e^{2\tau} + |c_1|^{2} e^{-2\tau}}
= -\,\frac{1 - 2\epsilon(\tau)}{1 + \epsilon(\tau)}
= -1 + \frac{2\,\epsilon(\tau)}{1+\epsilon(\tau)} \xrightarrow[\tau\to\infty]{} -1,
\end{align}
where we define the excess energy as
\begin{align}
    \epsilon(t) = \frac{|c_1|^2}{|c_0^2|} e^{-4t}
    = \frac{1 - |c_0|^2}{|c_0^2|} e^{-4t} = \left(\frac{1}{F_0} - 1\right) e^{-4t}\ .
\end{align}
\subsubsection{Energy gain after one DBAC step of one qubit}
\begin{lem}[Energy Change under One Step of DBAC for single-qubit]
Let $\ket{\psi_0}$ be a pure single qubit state with energy $E_0$.
Then 1 step of DBAC with duration $s$ yields a state with energy:
\begin{align}
    \frac{E_1}{2} = \frac{E_0}{2}-\sin^2(s)(1-E_0^2)\Big((1-\cos(s)) E_0 + \cos(s)\Big)\ ,\label{eq:E1-energy}
\end{align}
\end{lem}
\begin{proof}
We begin by setting notation.
Since DBAC is only concerned with the state energy (the $Z$ axis), without loss of generality, we can redefine $X$ and $Y$ axis to get
\begin{align}
\psi_0 = \frac{1}{2} \id + \alpha_X X + \alpha_Z Z\ ,\label{eq:psi-0-redefined-x-y}
\end{align}
where $\alpha_X,\alpha_Z \in \mathbb{R}$, to ensure that $\ket{\psi_0}$ is Hermitian.
Moreover,
\begin{align}
    E_0 = \text{tr}[H\psi_0]=-2\alpha_Z
\end{align}
and thus
without loss of generality
\begin{align}
\psi_0 = \frac{1}{2} \id + \alpha_X X -\frac{E_0}{2} Z\ ,\label{eq:psi-0-redefined-x-y}
\end{align}

DBAC is based on a recursion relation. We want to compute the state after $k=1$ step:
\begin{align}
    \psi_1 &= V_0 \psi_0 V_0^\dagger\ . \label{eq:define-recursion}
\end{align}

The evolution is governed by the unitary:
\begin{align}
    V_0 &= e^{isH} e^{is\psi_0} e^{-isH}
    = \id + (e^{i s} -1) e^{isH} \psi_0 e^{-isH}
    = \id + (e^{i s} -1) e^{-isZ} \psi_0 e^{isZ}
    =: \id + v(s)  \psi_0(s)\ ,  \label{eq:define-vs}
\end{align}
where $s$ is the step size, and the shorthands $v(s)= e^{is}-1$, $\psi_0(s) = e^{-isZ}\psi_0 e^{isZ}$ are used. Note that in the second equality we used the identity for pure states $e^{ is\psi} = \id + (e^{ is} -1)\psi$. Next, we explicitly evaluate
\begin{align}
    \psi_0(s)=e^{-isZ}\psi_0 e^{isZ}
    &=\frac{1}{2} \id + \alpha_X \cos(2s) X - \alpha_X \sin(2s) Y - \frac {E_0}2Z\\
    &=\psi_0 - \alpha_X (1-\cos(2s)) X + \alpha_X \sin(2s) Y \\
    &=\psi_0 + x_0 X +y_0 Y \\
    &=\psi_0 + \eta\ , \label{eq:define-eta}
\end{align}
where we define 
\begin{align}\label{eq:x0y0def}
x_0 = -\alpha_X (1-\cos(2s)), \qquad y_0 = \alpha_X \sin(2s),    
\end{align}
and also $\eta = x_0 X +y_0 Y$. Note that in the above calculation, we have made use of the identities
\begin{align}
    e^{-isZ} X e^{isZ} &= \cos(2s) X + \sin(2s) Y\ , \qquad
    e^{-isZ} Y e^{isZ} = \cos(2s) Y - \sin(2s) X\ .
\end{align}

Substitute Eq.~\eqref{eq:define-vs} and~\eqref{eq:define-eta} into Eq.~\eqref{eq:define-recursion}
\begin{align}
    \psi_1 &= V_0 \psi_0 V_0^\dagger\\
    &=(\id + v  \psi_0(s))\psi_0(\id + v^*  \psi_0(s))\\
    &= \psi_0 + v (\psi_0 + \eta) \psi_0  + v^*  \psi_0 (\psi_0 + \eta) +|v|^2 (\psi_0 + \eta) \psi_0  (\psi_0 + \eta)\\
    &= \psi_0 + (v + v^*) \psi_0^2 + v \eta \psi_0  + v^*\psi_0 \eta + |v|^2 (\psi_0^3 + \eta \psi_0^2 +\psi_0^2 \eta + \eta \psi_0 \eta)\ .
\end{align}
Again, we make use of the fact that for pure states $\ket{\psi_0}$, $\psi_0^k = \psi_0$ for all integers $k$. 
Next, we rearrange terms to get
\begin{align}
    \psi_1 =& (1 + v + v^* + |v|^2) \psi_0 
    + |v|^2 \{\eta ,\psi_0 \}
    + v \eta \psi_0   
    + v^*\psi_0 \eta 
    + |v|^2 (\eta \psi_0 \eta)\ .
    \label{eq:symbol_psi_1}
\end{align}

Now we calculate the terms involved in or will be relevant to calculate Eq. \eqref{eq:symbol_psi_1}:

% ... previous content unchanged ...

\begin{enumerate}
    \item Calculate $\eta \psi_0 \eta$:
        \begin{align}
        (x_0 X +y_0 Y)X  (x_0 X +y_0 Y) 
        &=(x_0^2 -y_0^2)X + 2x_0y_0 Y\\
        (x_0 X +y_0 Y)Y  (x_0 X +y_0 Y)
        &=(y_0^2 -x_0^2)Y + 2x_0y_0 X\\
        (x_0 X +y_0 Y)Z  (x_0 X +y_0 Y)
        &=-(x_0^2 +y_0^2)Z\ .
        \end{align}   
        \begin{align}
         \eta\psi_0 \eta&=(x_0 X +y_0 Y)\psi_0  (x_0 X +y_0 Y) \\  
        &=(x_0 X +y_0 Y)(\frac{1}{2} \id + \alpha_X X - \frac{E_0}{2}Z)  (x_0 X +y_0 Y) \\  
        &=\frac{1}{2} (x_0^2 + y_0^2) \id + \alpha_X [(x_0^2-y_0^2)X + 2x_0y_0 Y] +\frac{E_0}{2}(x_0^2 +y_0^2)Z \\
        &=\frac{1}{2} (x_0^2 + y_0^2) \id + \alpha_X(x_0^2-y_0^2)X + 2\alpha_Xx_0y_0 Y +\frac{E_0}{2}(x_0^2 +y_0^2)Z 
        \end{align}
    \item Calculate $\eta \psi_0$:
        \begin{align}
        \eta \psi_0 &= (x_0 X +y_0 Y)(\tfrac{1}{2} \id + \alpha_X X - \tfrac{E_0}{2}Z) \\
        & = x_0 X (\tfrac{1}{2} \id + \alpha_X X - \tfrac{E_0}{2}Z) + y_0 Y (\tfrac{1}{2} \id + \alpha_X X - \tfrac{E_0}{2}Z)\\
        & = x_0 (\tfrac{1}{2} X + \alpha_X \id + i\tfrac{E_0}{2}Y) + y_0(\tfrac{1}{2} Y - i\alpha_X Z - i\tfrac{E_0}{2}X)\\
        & = \alpha_X  x_0 \id
        + \left(x_0 \tfrac{1}{2} - y_0 i\tfrac{E_0}{2}\right)X
        + \left(\tfrac{1}{2}y_0+ x_0 i\tfrac{E_0}{2}\right) Y
        - i\alpha_X y_0 Z\ .
        \end{align}
    \item Calculate $\psi_0\eta$ by using $\psi_0^\dagger = \psi_0$ and $\eta^\dagger = \eta$:
        \begin{align}
            \psi_0\eta &= (\eta^\dagger\psi_0^\dagger)^\dagger = (\eta\psi_0)^\dagger\\
            & = \alpha_X  x_0 \id
            + \left(x_0 \tfrac{1}{2} + y_0 i\tfrac{E_0}{2}\right)X
            + \left(\tfrac{1}{2}y_0 - x_0 i\tfrac{E_0}{2}\right) Y
            + i\alpha_X y_0 Z\ .
        \end{align} 
    \item Combining items 2 and 3 allows us to calculate $\{\eta, \psi_0\}$:
        \[
        \eta \psi_0 + \psi_0 \eta = 2\alpha_X  x_0 \id + x_0 X + y_0 Y \, .
        \]
    \item Calculate $\alpha_X$ by using purity $\psi_0^2=\psi_0$:
        \begin{align}
            \psi_0^2 &= \left( \tfrac{1}{2} \id + \alpha_X X - \tfrac{E_0}{2} Z \right) \left( \tfrac{1}{2} \id + \alpha_X X - \tfrac{E_0}{2} Z \right)\\
            &= \left( \tfrac{1}{4} + \alpha_X^2 + \tfrac{E_0^2}{4} \right) \id + \alpha_X X - \tfrac{E_0}{2} Z\ . \label{eq:alpha-x-calculation-intermediate}
        \end{align}
        Equating Eq.~\eqref{eq:alpha-x-calculation-intermediate} with Eq.~\eqref{eq:psi-0-redefined-x-y}, we get
        \begin{align}
            \alpha_X^2 &= \tfrac{1}{4} (1 - E_0^2)\ .
        \end{align}
    \item Calculate $x_0^2 + y_0^2$: 
        \[
        \begin{aligned}
        x_0^2 + y_0^2 &= \alpha_X^2 \!\left( 1 - 2\cos(2s) + \cos^2(2s) + \sin^2(2s) \right) \\
        &= 2\alpha_X^2 \!\left(1-\cos(2s)\right) \\
        &= -2\alpha_X x_0 \, .
        \end{aligned}
        \]
\end{enumerate}

Grouping all the terms:
\begin{align}
    \psi_1 =& (1 + v + v^* + |v|^2) \psi_0 
    + |v|^2 \{\eta \psi_0 + \psi_0 \eta\}+ v \eta \psi_0   
    + v^*\psi_0 \eta 
    + |v|^2 (\eta \psi_0 \eta)\\
    =& |1+v|^2 \psi_0 
    + |v|^2 \{\eta \psi_0 + \psi_0 \eta\}+ v \eta \psi_0   
    + v^*\psi_0 \eta 
    + |v|^2 (\eta \psi_0 \eta)\\
    =& |1+v|^2\left(\frac{1}{2} \id + \alpha_X X - \frac{E_0}{2}Z\right)\\
    &+|v|^2 \left(2\alpha_X  x_0 \id + x_0 X + y_0 Y\right)\\
    &+ v \left[\alpha_X  x_0 \id
    + \left(x_0 \frac{1}{2} - y_0i\frac{E_0}{2}\right)X
    + \left(\frac{1}{2}y_0+ x_0i\frac{E_0}{2}\right) Y
    -i\alpha_Xy_0 Z\right]\\
    &+v^*\left[\alpha_X  x_0 \id
    + \left(x_0 \frac{1}{2} + y_0i\frac{E_0}{2}\right)X
    + \left(\frac{1}{2}y_0 - x_0i\frac{E_0}{2}\right) Y
    +i\alpha_Xy_0 Z\right]\\
    &+|v|^2 \left[\frac{1}{2} (x_0^2 + y_0^2) \id + \alpha_X(x_0^2-y_0^2)X + 2\alpha_Xx_0y_0 Y +\frac{E_0}{2}(x_0^2 +y_0^2)Z \right]\ .
\end{align}

We note that $X$ and $Y$ terms are irrelevant to our task of resetting qubit, because we are only interested in the energy of the state, which depends on $\id$ and $Z$ terms. Hence we regroup to get
\begin{align}
    \psi_1=&\left[\frac{1}{2}|1+v|^2 + |v|^2 2\alpha_X  x_0 + (v+v^*)\alpha_X  x_0 + |v|^2 \frac{1}{2} (x_0^2 + y_0^2)\right]\id\\
    &+C_X X + C_Y Y+\left[-|1+v|^2\frac{E_0}{2} 
    + (v^* - v)i\alpha_Xy_0
    + |v|^2 \frac{E_0}{2}(x_0^2 +y_0^2)
    \right]Z\\
    =&\left[\frac{1}{2}|1+v|^2 + |v|^2 \alpha_X  x_0 + (v+v^*)\alpha_X  x_0\right]\id\\
    &+C_X X + C_Y Y
    +\left[-|1+v|^2\frac{E_0}{2} 
    + (v^* - v)i\alpha_Xy_0
    - |v|^2 \alpha_X x_0E_0
    \right]Z\ .
\end{align}

Now, we use $v(s) = e^{is} - 1$ to simplify $|1+v|^2 = 1$ and $v+v^* + |v|^2 = 0$. We then use the definition for $x_0$ and $y_0$ in Eq.~\eqref{eq:x0y0def} to arrive at
\begin{align}
    \psi_1=&\frac{1}{2}\id+C_X X + C_Y Y+\left[-\frac{E_0}{2} 
    + (-2i\sin(s))i\alpha_Xy_0
    - 2(1-\cos(s)) \alpha_X x_0E_0
    \right]Z\\
    =&\frac{1}{2}\id+C_X X + C_Y Y+\left[-\frac{E_0}{2} 
    -2\alpha_X^2\Big(-(1-\cos(s))(1-\cos(2s)) E_0 + \sin(s) \sin(2s)\Big)
    \right]Z\\
    =&\frac{1}{2}\id+C_X X + C_Y Y+\left[-\frac{E_0}{2} 
    -2\alpha_X^2\Big(-(1-\cos(s))2\sin^2(s) E_0 + 2\sin^2( s) \cos(s)\Big)
    \right]Z\\
    =&\frac{1}{2}\id+C_X X + C_Y Y+\left[-\frac{E_0}{2} 
    -4\sin^2(s)\alpha_X^2\Big((1-\cos(s)) E_0 + \cos(s)\Big)
    \right]Z\\
    =&\frac{1}{2}\id+C_X X + C_Y Y-\frac12\left[E_0-2\sin^2(s)(1-E_0^2)\Big((1-\cos(s)) E_0 + \cos(s)\Big)
    \right]Z  \\
    \equiv&\frac{1}{2}\id+C_X X + C_Y Y-\frac{E_1}2Z
    \ .\label{eq:energy_change}
\end{align}

The energy of the state $\ket{\psi_1}$ is thus given by Eq.~\eqref{eq:E1-energy}.
\end{proof}

\subsection{DBAC Energy bound for more qubit cases}

We show how DBAC can be used for resetting a qubit.
Let $U_k$ denote the circuit to prepare $\ket{\omega_k}$ from a trivial reference state $\ket 0$, i.e., $\ket{\omega_k} = U_k \ket{0}$. 
We can now use unitarity to simplify $e^{i\sqrt{s}\omega_k}= U_k e^{i\sqrt{s_k}\ket 0\langle0|}    U_k^\dagger$ and arrive at a recursive formula for DBAC circuit synthesis:
    \begin{align}
    U_{k+1} = e^{i\sqrt{s_k}\h}U_k e^{i\sqrt{s_k}\ket 0\langle0|}    U_k^\dagger
  e^{-i\sqrt{s_k}\h} U_k\ \ .
  \label{DBAC Uk}
\end{align}
The average energy $E_k:=\bra{\omega_k}\h\ket{\omega_k}$ of the states $\ket{\omega_k}:= U_k \ket 0$, where $U_k$ is defined recursively in Eq.~\eqref{DBAC Uk}, are bounded as~\cite{dbqite}
\begin{align}
     E_{k+1} \le  E_k - 2s_k V_k + \mathcal O (s_k^{2})
    \label{eq fluctuation-refrigeration main Os2} \ 
\end{align}
where $V_k:=\bra{\omega_k}\h^2\ket{\omega_k}-E_k^2$ is the variance of the energy in state  $\ket{\omega_k}$. This demonstrates that the energy of the qubit decreases (or at least remains constant) at each step.

\section{Numerics}
\subsection{Comparing DBAC performance with ideal DME}

We benchmark DBAC with ideal DME (noise-free $\delta$-SWAP). Figure~\ref{fig:ideal_DME} quantifies the gain in ground-state fidelity as a function of the number of DB-QITE steps $k$, and for various step sizes $s$ and instruction depths $M$. The curves converge within a few steps and identify the minimum initial fidelity needed to reach a target (e.g., $F=0.9$).

\begin{figure}[H]
    \centering
    \includegraphics[width=1.0\linewidth]{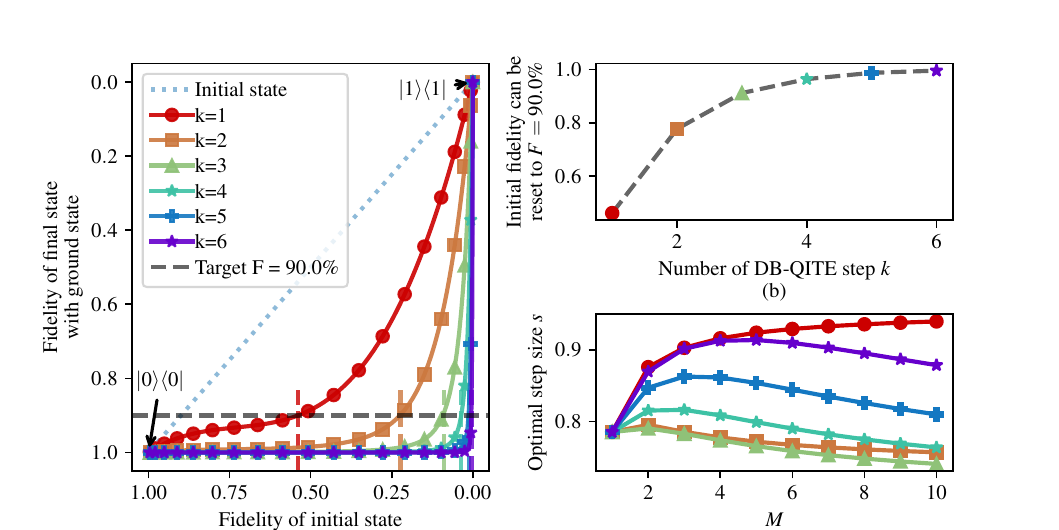}
    \caption{DBAC performance for different number of steps $k$, if DME is ideal. Fig. (a) shows that after $k=6$ steps all the state can be reset, except for the state $\ket 1\bra 1$. Fig. (b) shows the minimum initial fidelity  of initial state to the ground state (that is, how far away the initial state can be from the ground state) that can still be reset to a final fidelity $F=90\%$. Fig. (c) shows optimal step size $s$ for different number of steps $k$ and $M$. For different combination of $k$ and $M$, the optimal step size $s$ that lower the energy of the qubit the most is different.}
    \label{fig:ideal_DME}
\end{figure}

\subsection{Bloch sphere}
Figure~\ref{fig:bloch_combined} visualizes representative DBAC trajectories on the Bloch sphere for different initial overlaps with the ground state. The examples illustrate how increasing $k$ and $M$ enlarges the basin of successful reset to the target fidelity. Conversely, states with very small initial overlap remain outside the attainable region for the $(k,M)$ values shown.

\begin{figure}[H]
    \centering
    % \begin{subfigure}{0.45\textwidth}
    \begin{subfigure}{0.3\textwidth}
        \centering
         \includegraphics[width=0.65\linewidth]{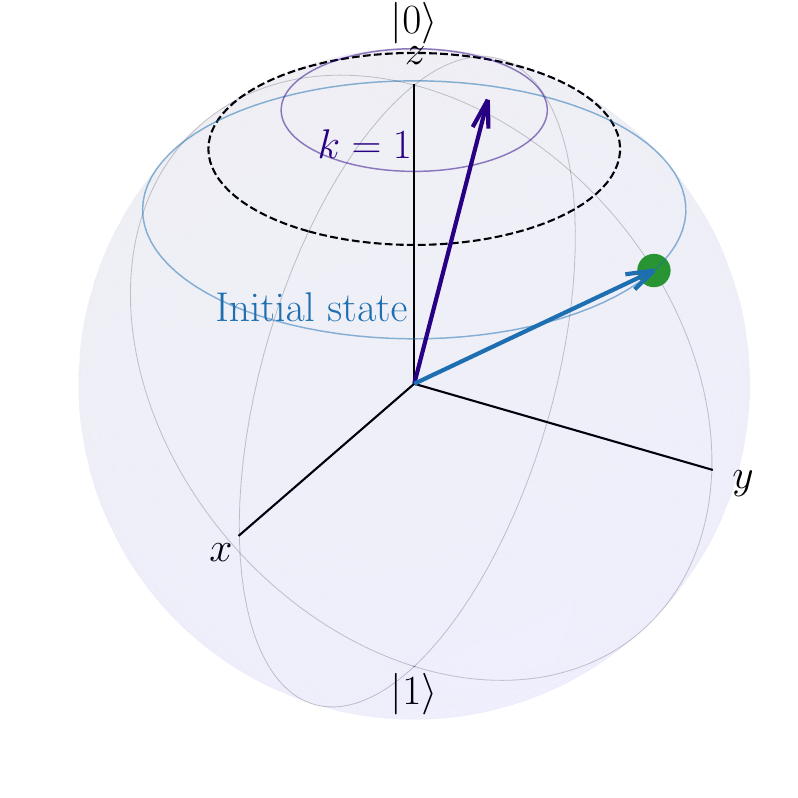}
        \caption{Initial state close to ground state $F_{\text{initial}}=80\%$}
        \label{fig:blochgood}
    \end{subfigure}
    % \hfill
    \begin{subfigure}{0.3\textwidth}
        \centering
        \includegraphics[width=0.65\linewidth]{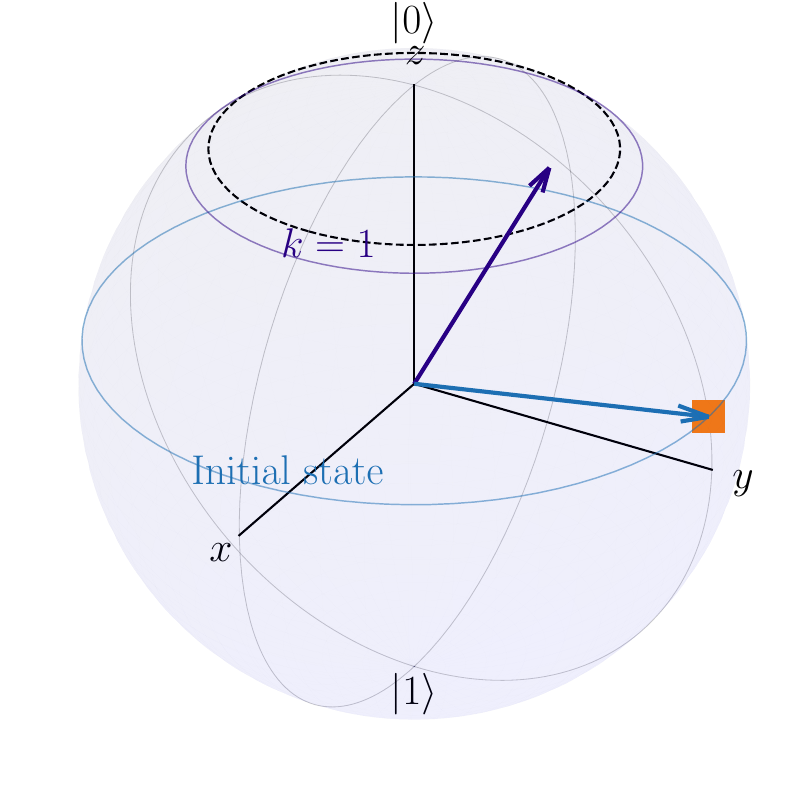}
        \caption{Initial state not close to ground state $F_{\text{initial}}=60\%$}
        \label{fig:blochmid}
    \end{subfigure}
    % \hfill
    \begin{subfigure}{0.3\textwidth}
        \centering
        \includegraphics[width=0.65\linewidth]{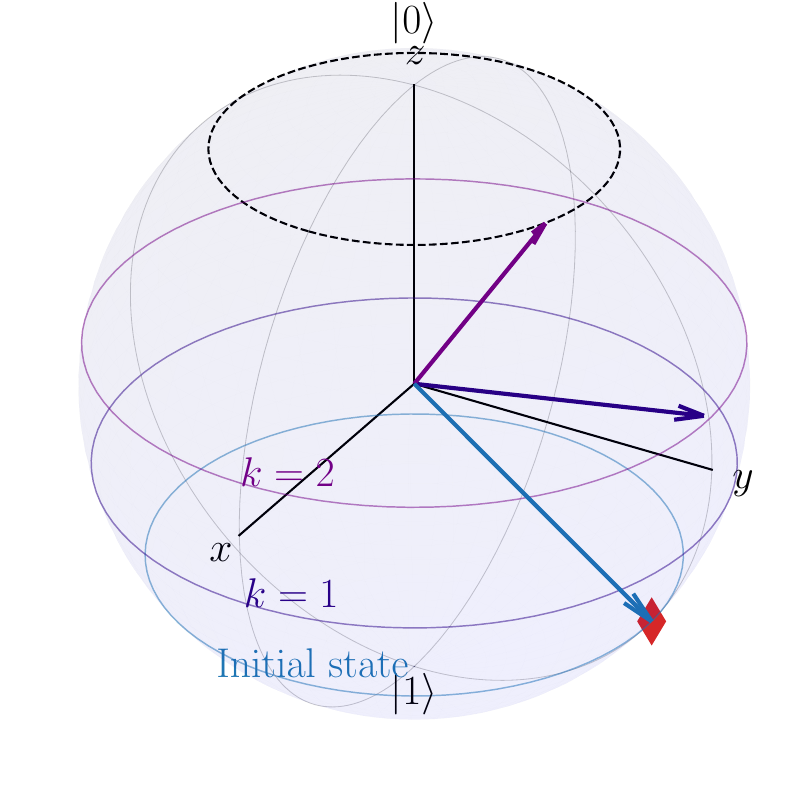}
        \caption{Initial state far from ground state $F_{\text{initial}}=10\%$}
        \label{fig:blochbad}
    \end{subfigure}
    \caption{Illustration of DBAC for different initializations, comparing cases where the initial state is close to or far from the ground state. The dotted black latitude line represents the target fidelity with the ground state, $F_{\text{target}}=0.9$. Each latitude line corresponds to the state indicated by the arrow of the same color, showing both the initial state and its evolution after $k=1$ or $k=2$ steps.
    (a) The initial state has a fidelity of $F_{\text{initial}}=80\%$ with the ground state. The target fidelity is achieved after $k=1$ step with $M=1$, requiring only one extra copy of the initial state.
    (b) The initial fidelity with the ground state is $F_{\text{initial}} = 60\%$. The target fidelity is reached after $k=2$ steps with $M=2$, which consumes eight extra copies of the initial state.
    (c) The initial state is far from the ground state, with $F_{\text{initial}} = 10\%$. In this case, the target fidelity cannot be reached within $k=2$ steps and $M=2$.}
    \label{fig:bloch_combined}
\end{figure}

% Analytical calculations DBAC qubit reset

\subsubsection{Step size s~ sweep}
Figure~\ref{fig:s_sweep} sweeps the step size $s$ versus the initialization angle $\theta$. In both $k=1$ cases ($M=1$ and $M=2$), there exists a state-dependent threshold $s^{*}$ beyond which the final fidelity meets the target over a broad range of $\theta$. This indicates robustness to moderate over-rotation in $s$ and complements the optimal values extracted in Figure~\ref{fig:ideal_DME}.

\begin{figure}[H]
    \centering
    % Image 1
    \begin{subfigure}{0.49\textwidth}
        \centering
        \includegraphics[width=\linewidth]{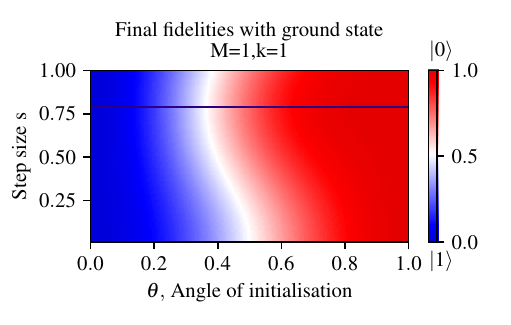}
        \caption{DBAC with $k=1, M=1$ step.}
        \label{fig:M1k1_s_sweep}
    \end{subfigure}
    % Image 2
    \begin{subfigure}{0.49\textwidth}
        \centering
        \includegraphics[width=\linewidth]{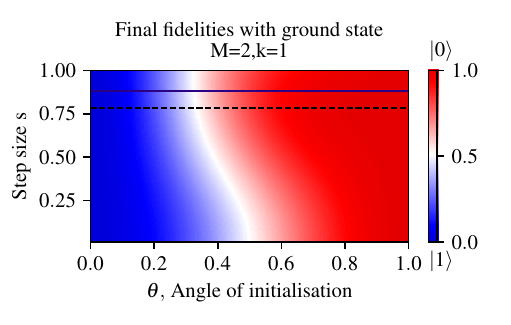}
        \caption{DBAC with $k=1, M=2$ steps.}
        \label{fig:M2k1_s_sweep}
    \end{subfigure}
    \caption{(a) and (b) show the step size $s$ sweep for the cases $k=1$ with $M=1$ and $M=2$. As long as the experimental step size s exceed the optimal s, DBAC will reset the qubit.}
    \label{fig:s_sweep}
\end{figure}

\section{Experiment set-up}
\subsection{Device layout and parameters}

The detailed device parameters for the set of qubits we consider in this work are given in Table. \ref{basic_device_parameters} and shown in in Fig. \ref{fig:layout} \cite{alghadeer2025low, alghadeer2025characterization}. Randomized benchmarking (RB) ~\cite{chow2009randomized, gambetta2012characterization} experiments were conducted using an \( XY \)-Clifford decomposition for obtaining simultaneous 4-qubits gate fidelities.

\begin{figure}[H]
    \centering
    \includegraphics[width=0.64\linewidth]{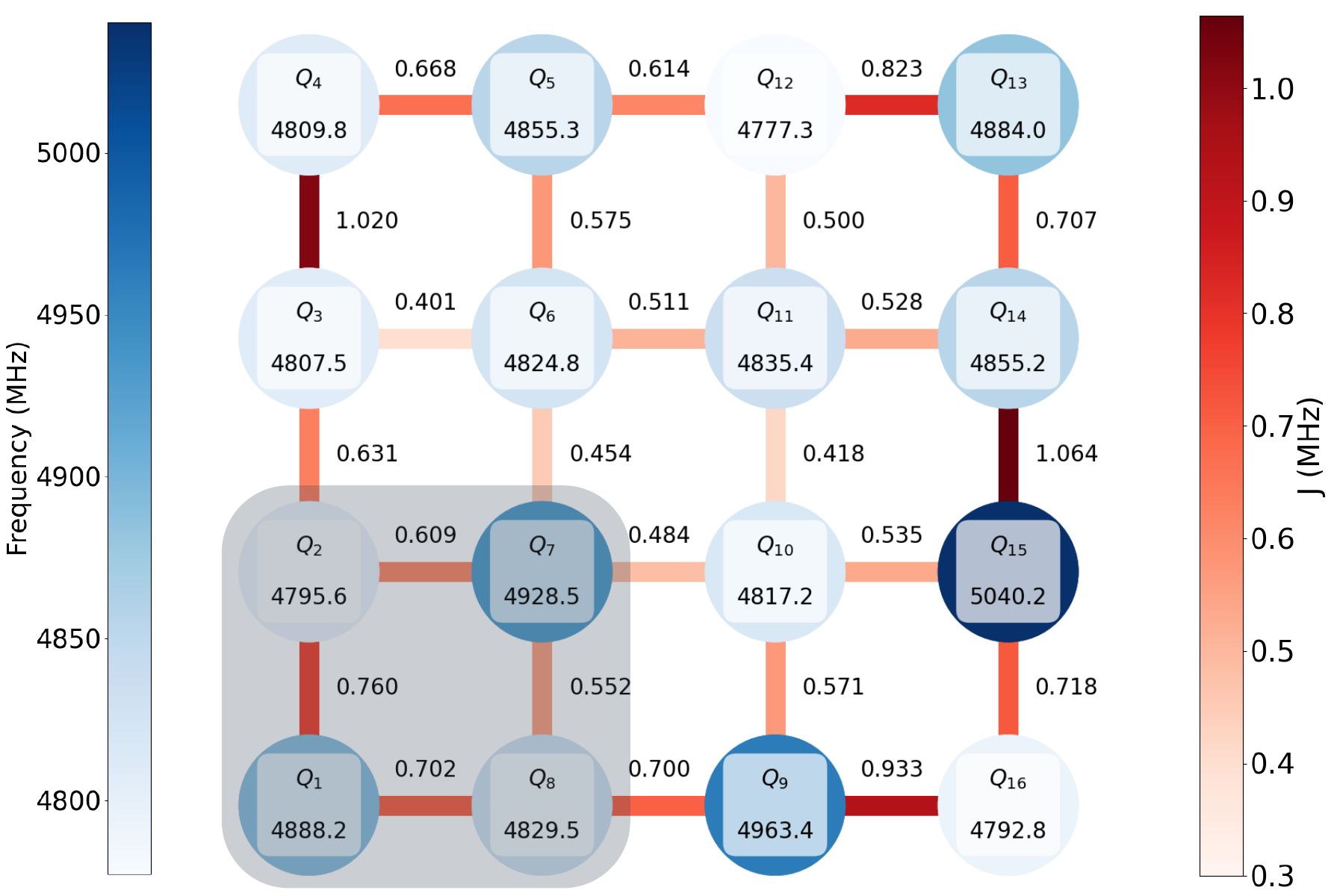}
    \caption{Layout of the 16-qubit chip. The four qubits used in the circuits from Table~\ref{tab:k_m_values_M1} are highlighted by the shadow square in the lower-left corner. The DME and DBAC protocols can be applied to other qubits or extended to larger subsets of the device. Here, $Q_{1}$, $Q_{2}$, $Q_{7}$ and $Q_{8}$ correspond to $q_{1}$, $q_{2}$, $q_{3}$ and $q_{4}$ in the circuit diagrams.}
    \label{fig:layout}
\end{figure}

\begin{table}[H]
    \centering
    \caption{Basic Device Parameters. $Q_{\text{int}}$ is internal Q factor of resonator, $\kappa_{\text{ext}}$ is external decay rate of resonator. Qubits relaxation and coherence times $T_{1}$, $T_{2R}$ and $T_{2E}$ are averaged over 400 repeated measurements. Simultaneous single-qubit gate fidelities are also given measured through randomized benchmarking.}
    \resizebox{0.8\textwidth}{!}{ % Resize the table to fit within the page width
        \begin{tabular}{|c|c|c|c|c|c|c|c|c|c|c|}
            \hline
            Parameters & \( w_{r}/2\pi \) & \( w_{q}/2\pi \) & \( Q_{i} \) & \( k_{ext} \) & \( \chi \) & \( \alpha \)  & \( \langle T_{1} \rangle \) & \( \langle T_{2R} \rangle \) & \( \langle T_{2E} \rangle \) & \( \mathcal{F} \) \\
            \hline
            Qubits & MHz & MHz & \( 10^{4} \) & MHz & KHz  & MHz & \( \mu s \) & \( \mu s \) & \( \mu s \) &  \%  \\
            \hline
            \( Q_{1} \)  & 9997.4 & 4888.2 & 11.7 & 2.6 & -200.0 & -196.6 & 126 ± 18 & 107 ± 12 & 124 ± 23 & 99.961  \\
            \( Q_{2} \)  & 9386.0 & 4795.6 & 11.8 & 1.3 & -225.0 & -197.2 & 89 ± 13  & 56 ± 15  & 86 ± 12  & 99.902  \\
            \( Q_{7} \)  & 9474.2 & 4928.5 & 10.8 & 1.7 & -175.0 & -195.6 & 77 ± 12  & 45 ± 12  & 82 ± 15  & 99.946  \\
            \( Q_{8} \)  & 9908.6 & 4829.5 & 5.8  & 2.0 & -175.0 & -197.2 & 63 ± 8   & 32 ± 7   & 71 ± 8   & 99.936  \\
            \hline
        \end{tabular}
    }
    \label{basic_device_parameters}
\end{table}

\subsection{Additional experimental results}
% \twocolumngrid
\onecolumngrid

\subsubsection{Two-qubit Native ZZ interaction: calibration procedure}

Among quantum computing architectures, superconducting transmon qubits~\cite{transmon} provide a compelling platform, as lithographically defined circuit elements naturally enable capacitive couplings between neighbouring qubits~\cite{alghadeer2025characterization,place2021new}. These couplings give rise to fast, high-fidelity native two-qubit operations~\cite{yan2018tunable}, which are essential for implementing multiple quantum protocols. Although residual inter-qubit coupling can introduce correlated errors and always-on state-dependent ZZ shifts~\cite{ketterer2023characterizing,tripathi2022suppression,murali2020software,krinner2020benchmarking,zhao2022quantum,fors2024comprehensive}, careful engineering both suppresses detrimental interactions and repurposes the remaining couplings as a computational resource~\cite{alghadeer2025low}. In this work, we implement DBAC on a superconducting quantum lattice with well-controlled nearest-neighbour ZZ interactions, enabling the precise, programmable two-qubit operations required to redistribute entropy locally and efficiently.

We implement entangling operations between fixed-frequency transmon qubits in the lattice by using the Stark-induced ZZ by level excursions (siZZle) technique \cite{sizzlePRL, PhysRevLettMitchell} to boost static ZZ coupling. We use two additional off-resonant drives to induce parametrized shifts in the energy levels of a two-transmons system as shown in Fig. ~\ref{siZZle_Calibration}, where we perform a two-dimensional parameter sweep over the Stark drive frequency and amplitude space in order to building a coarse map of the ZZ interaction rates. This approach modifies the native ZZ interaction and can be used to tune up a controlled-Z (CZ) gate by driving each transmon with a detuned microwave tone. The two simultaneous off-resonant drives on the two qubits shift the energy levels of the system and, through the capacitive coupling, modify the effective ZZ rate between the qubits. The modified ZZ rate $\tilde{\nu}_{ZZ}$ can be approximated as \cite{sizzlePRL}:

\begin{equation}
\tilde{\nu}_{ZZ} = \nu_{ZZ,s} + \frac{2J \alpha_0 \alpha_1 \Omega_0 \Omega_1 \cos(\phi_0 - \phi_1)}{\Delta_{0,d} \Delta_{1,d} (\Delta_{0,d} + \alpha_0)(\Delta_{1,d} + \alpha_1)},
\label{eq:ZZ_sizzle}
\end{equation}

\begin{equation}
\nu_{ZZ,s}  \approx -\frac{2J_{ij}^2(\alpha_i + \alpha_j)}{(\Delta_{ij} + \alpha_i)(\alpha_j - \Delta_{ij})},
\label{eq:ZZ}
\end{equation}
where $\nu_{ZZ,s}$ denotes the static ZZ interaction rate, $\alpha_0$ and $\alpha_1$ are the anharmonicities, $J$ is the fixed capacitive coupling strength, and $\Omega_0$, $\Omega_1$, $\Delta_{0,d}$, $\Delta_{1,d}$, $\phi_0$ and $\phi_1$ denote the drive amplitudes, detunings to each qubit, and relative drive phases, respectively. In principle, the effective $\tilde{\nu}_{ZZ}$ rate can be boosted or canceled depending on the sign of the additional driving term. The calibration of a CZ gate based on the siZZle interaction requires optimizing the drive parameters such that the total ZZ-induced phase accumulation during the gate operation equals $\pi/4$.

Pulse sequence for siZZle gate calibration is shown in Fig. ~\ref{siZZle_Calibration}(a). The calibration of siZZle requires optimizing the drive parameters such that the total ZZ-induced phase accumulation during the gate operation equals $\pi/4$. This involves first tuning up the frequencies of the off-resonant drives to achieve optimal detunings $\Delta_{0,d}$ and $\Delta_{1,d}$ from the qubits' transitions, selecting both optimal amplitudes ($\Omega_0$ and $\Omega_1$) and phase difference ($\Delta = \phi_0 - \phi_1)$ between the off-resonant drives to maximize the $\tilde{\nu}_{ZZ}$ rate, and finally working the gate duration out of the optimal $\tilde{\nu}_{ZZ}$ rate. 

We set the two drive amplitudes to be relatively equal to $\Omega_{control} = r ~\Omega_{target}$, for maximum $\tilde{\nu}_{ZZ}$ while observing a clean interaction, and $r$ here is the ratio between the amplitudes of the single-qubit $X_{\pi}$ pulses for the control to the target qubits. We have found this relation takes into account the asymmetries between the two drive amplitudes, introduced by cabling or room-temperature electronics. The relative phase between the two drives was found to be near-optimal in our setup and was set to be $\Delta = \phi_0 - \phi_1 = 0$ during CZ gate calibration (See supplementary materials for more details).

Fine calibrated native ZZ interactions from the basis of the algorithm and each interaction is more precisely measuring the ZZ rate at every point close toe the course calibrated frequency obtained from the 2D map shown in Fig. ~\ref{siZZle_Calibration}. We then fix the Stark pulse duration at $1~\mu$s and use the Hamiltonian tomography sequence in Fig. ~\ref{siZZle_Tuneup}, recording the differential phase accumulation on the target qubit on $\langle Y \rangle$ basis when the control qubit is initialized in either the ground or excited state. This entire mapping procedure takes roughly 12 hours and yields the background for the interaction landscape over a wide range of parameters. Following this, we analyze both the target qubit (in Fig. ~\ref{siZZle_Calibration}(a)) and control qubit (in Fig. ~\ref{siZZle_Calibration}(b)) dynamics to identify the optimal parameters for high-contrast, coherent interactions. Next, we perform fine calibration by sweeping the measured ZZ rate as function of the drive amplitude (Fig. ~\ref{siZZle_Fine_Tuneup}(a)) and drive frequency (Fig. ~\ref{siZZle_Fine_Tuneup}(b)).

\begin{figure*}[t]
  \centering
  \includegraphics[width=1.0\textwidth]{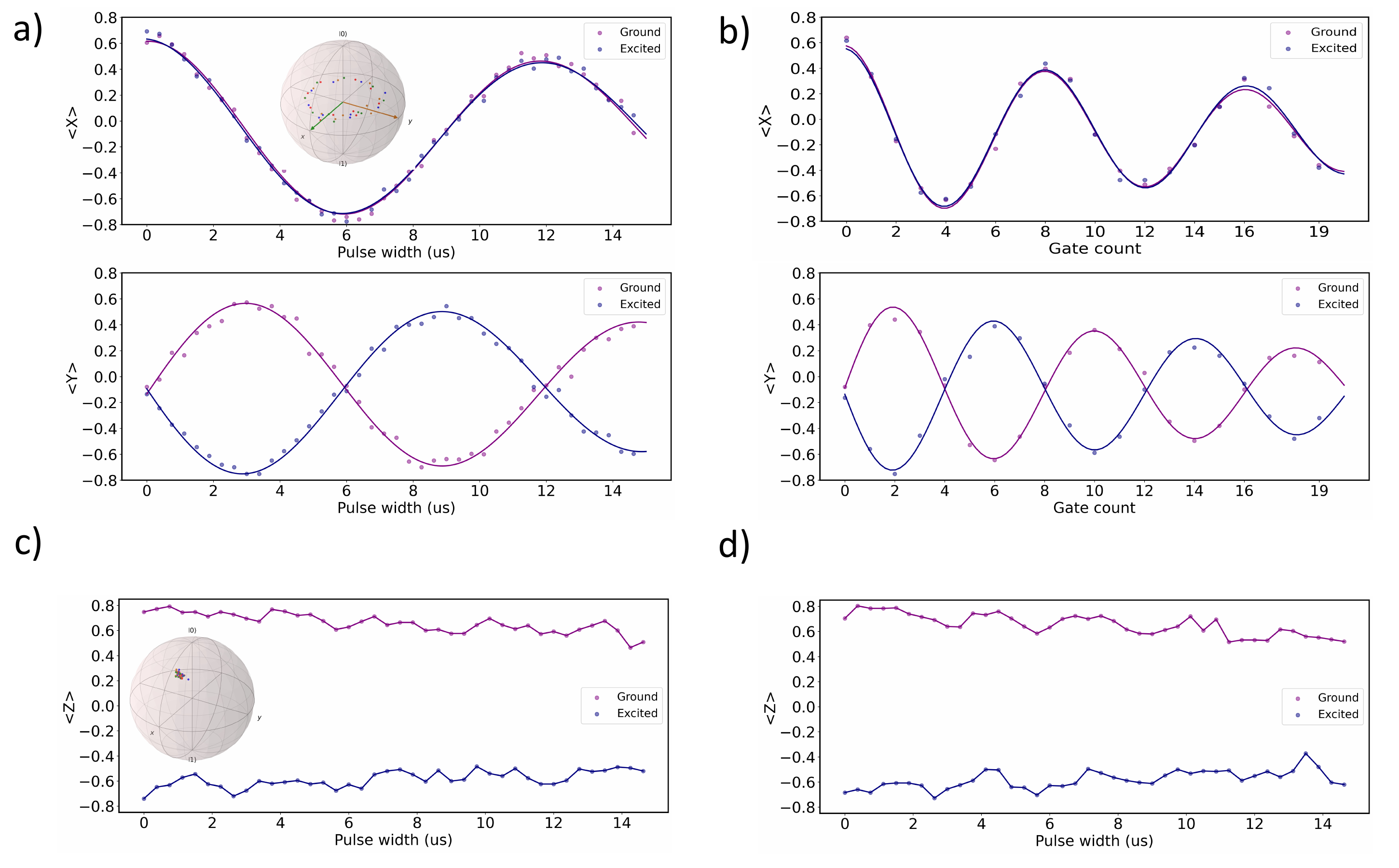}
  \caption{
        Two-qubit ZZ-induced phase accumulation on target qubit for calibrating a siZZle interaction on $Q_{2}$ (control) and $Q_{7}$ (target) as function of AC pulse width in (b) and as function of CZ gate count in (c) after setting up an optimal gate duration. States of the control qubit $Q_{2}$ during ZZ-induced phase accumulation on target qubit $Q_{7}$ during  $\langle X \rangle$ measurements in (d) and $\langle Y \rangle$ measurements in (e), for calibrating a siZZle gate between $Q_{2}$ and $Q_{7}$.}
  \label{siZZle_Tuneup}
\end{figure*}

\begin{figure}[H]
  \centering
   \includegraphics[width=0.9\textwidth]{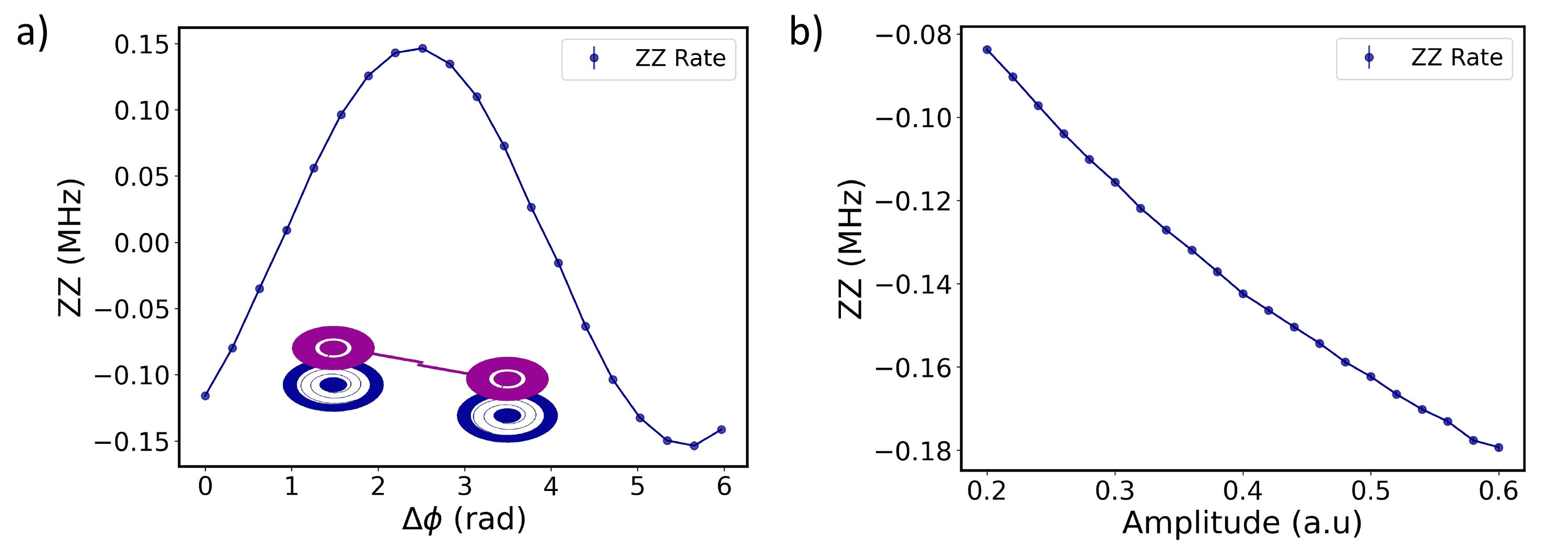}
  \caption{(a) Measured modulation of the ZZ interaction rate as a function of the relative phase difference $\Delta = \phi_0 - \phi_1$ between the two off-resonant drives. (b) ZZ interaction strength as a function of equal-amplitude off-resonant drives applied to both qubits, used to extract the optimal amplitude.}
  \label{siZZle_Fine_Tuneup}
\end{figure}

The calibration procedure continues by selecting an initial set of Stark drive parameters. Specifically, a drive frequency and amplitude are first selected to observe an initial ZZ interaction. This assessment is carried out using Hamiltonian tomography followed by repeated gate'tomography, as illustrated in Fig. ~\ref{siZZle_Tuneup}(a) and (b), respectively. Non-optimal parameters typically lead to unstable or noisy oscillations in the expectation values. In such cases, the parameters are iteratively adjusted until clean, stable oscillations are observed. Once stability is achieved, an automated fitting routine is used to extract the optimal gate duration to proceed to the calibration of a CZ entangling gate used for tuning up a CNOT gate for preparing Bell states \cite{cao2024agents, alghadeer2024psitrum, alghadeer2022psitrum}.

The siZZle gates sequence is implemented using two off-resonant Stark drives, with interleaved and final single-qubit $\pi$ pulses to cancel unwanted single-qubit phase accumulation as shown in Fig. ~\ref{siZZle_Calibration}(a). In Fig. ~\ref{siZZle_Tuneup}(a), Pulse width Hamiltonian tomography on the target qubit is shown to extract the ZZ interaction strength: the duration of the Stark pulse is swept while monitoring the phase evolution of the target qubit. A dashed line on the the single-qubit pulse on the control in Fig. ~\ref{siZZle_Calibration}(a) denotes that the experiment is run both with and without exciting the control. In Fig. ~\ref{siZZle_Tuneup}(b), repeated gate Hamiltonian tomography on the target qubit is shown: this experiment now uses fixed-duration two-qubit pulses and repetition blocks to more precisely calibrate the accumulated ZZ phase, which can then be used to implement an entangling gate. States of the control qubit $Q_{2}$ during ZZ-induced phase accumulation on target qubit $Q_{7}$ are shown in in Fig. ~\ref{siZZle_Tuneup}(c) and (d).

\begin{figure}[H]
  \centering
  \includegraphics[width=1.0\textwidth]{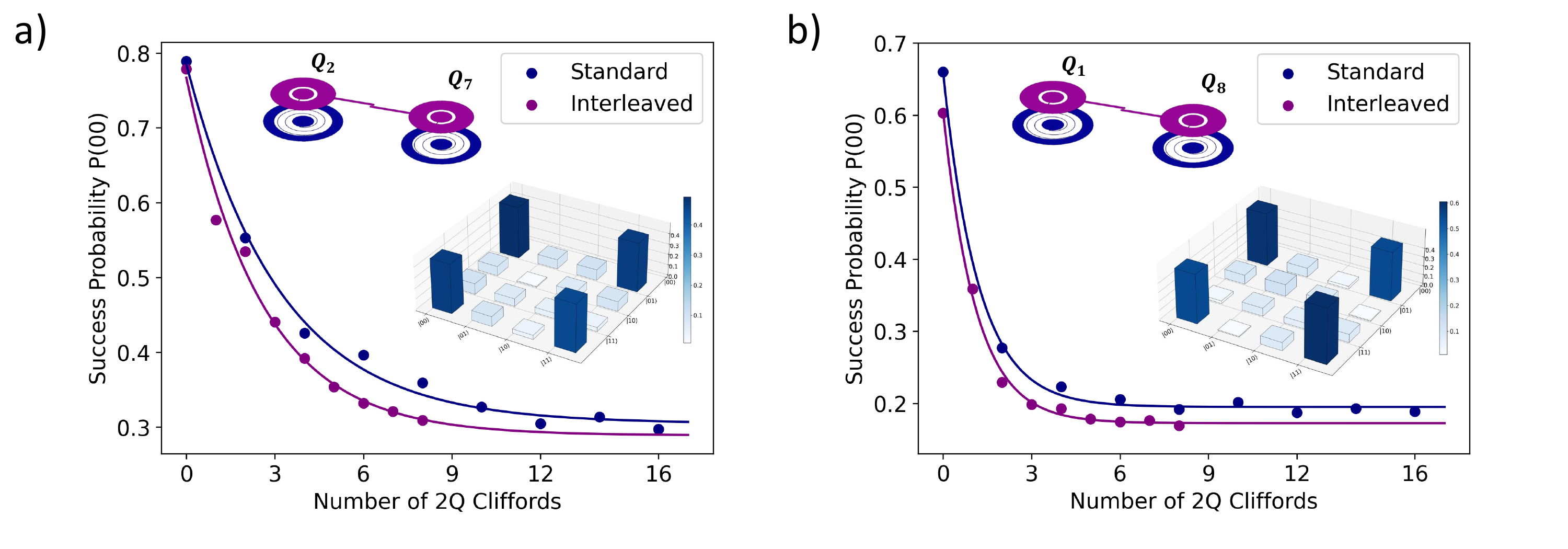}
  \caption{
        Two-qubit randomized benchmarking and Bell state preparation between $Q_{2}$ and $Q_{7}$ in (a), and between $Q_{1}$ and $Q_{8}$ in (b), demonstrating a two-qubit entanglement in the lattice.}
  \label{siZZle_Gate}
\end{figure}

We verified the calibrated siZZle technique by interleaving it into a two-qubit randomized benchmarking (RB) sequence, performed on two pairs as shown in Fig. ~\ref{fig:layout}. The single-qubit gates here were optimally calibrated with 20~ns duration (consisting of 16~ns Blackman envelope with 4~ns buffer) of $X_{\pi/2}$ physical gates, combined with derivative removal gate (DRAG) pulse shaping \cite{motzoi2009simple} and virtual Z gates \cite{mckay2017efficient}. Single-shot readout was also optimized during all two-qubit RB experiments with a readout time of 3~us. The performance of the CZ gates was further complemented by the direct preparation of entangled Bell states between two qubits (in Fig. ~\ref{siZZle_Gate}(a) and (b)). For example, the first pair in Fig. ~\ref{siZZle_Gate}(a) consisting of $Q_{2}$ (control) and $Q_{7}$ (target), we achieve CZ gate fidelity of $95.15 \pm 1.76~\%$ for a total gate time of $\tau_{g} = 3.266 ~\mu s$, resulting in an average Bell state fidelity of $93.56\%$ measured by two-qubit state tomography. For the second pair in Fig. ~\ref{siZZle_Gate}(b) consisting of $Q_{1}$ (control) and $Q_{8}$ (target), we achieve CZ gate fidelity of $96.44 \pm 1.78~\%$ for a total gate time of $\tau_{g} = 2.623 ~\mu s$, resulting in an average Bell state fidelity of $90.0\%$ measured by two-qubit state tomography. 

%%%%%%%%%%%%%%%%%%%%%%%%%%%%%%%%%%%%%%%%%%%%%%%%%%%%%%%%%%%%%%%%%%%%%%%%%%%%%%%%%%%%%%%%%%%%%

% \input{ptm}

% \newpage
% \clearpage
% \onecolumngrid
\subsubsection{Heisenberg interactions and DME implementation}

The figures in this section summarize the calibration and benchmarking of native Heisenberg-type interactions across selected qubit pairs. We present measured Pauli transfer matrices (PTMs) for ZZ, XX, and YY couplings on pairs QA–QB and QA–QD, directly compared with analytic predictions. 

% % ZZ, XX and YY for QA-QB
\begin{figure}[H]
    \centering
    \includegraphics[width=0.92\linewidth]{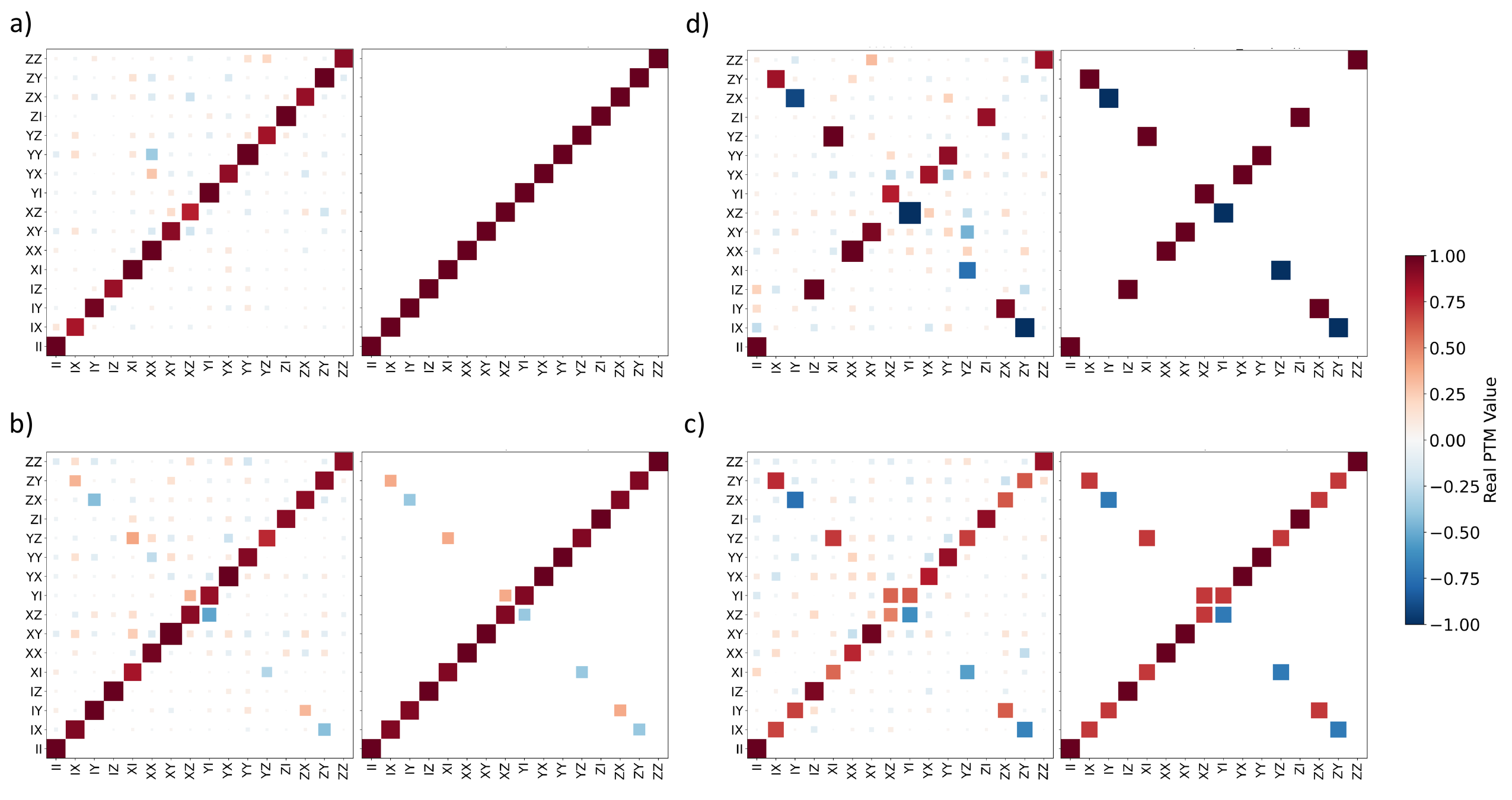}
    \caption{Measured and analytic Pauli transfer matrix (PTM) results of ZZ interaction on pair QA-QB at angle a) $\phi=0$, b) $\phi=\pi/8$, c) $\phi=\pi/4$ and d) $\phi=\pi/2$. The average process fidelity of each is a) $95.96\%$, b) $99.08\%$, c) $91.48\%$ and d) $96.98\%$.}
    \label{ZZ_AB}
\end{figure}

\begin{figure}[H]
    \centering
    \includegraphics[width=0.92\linewidth]{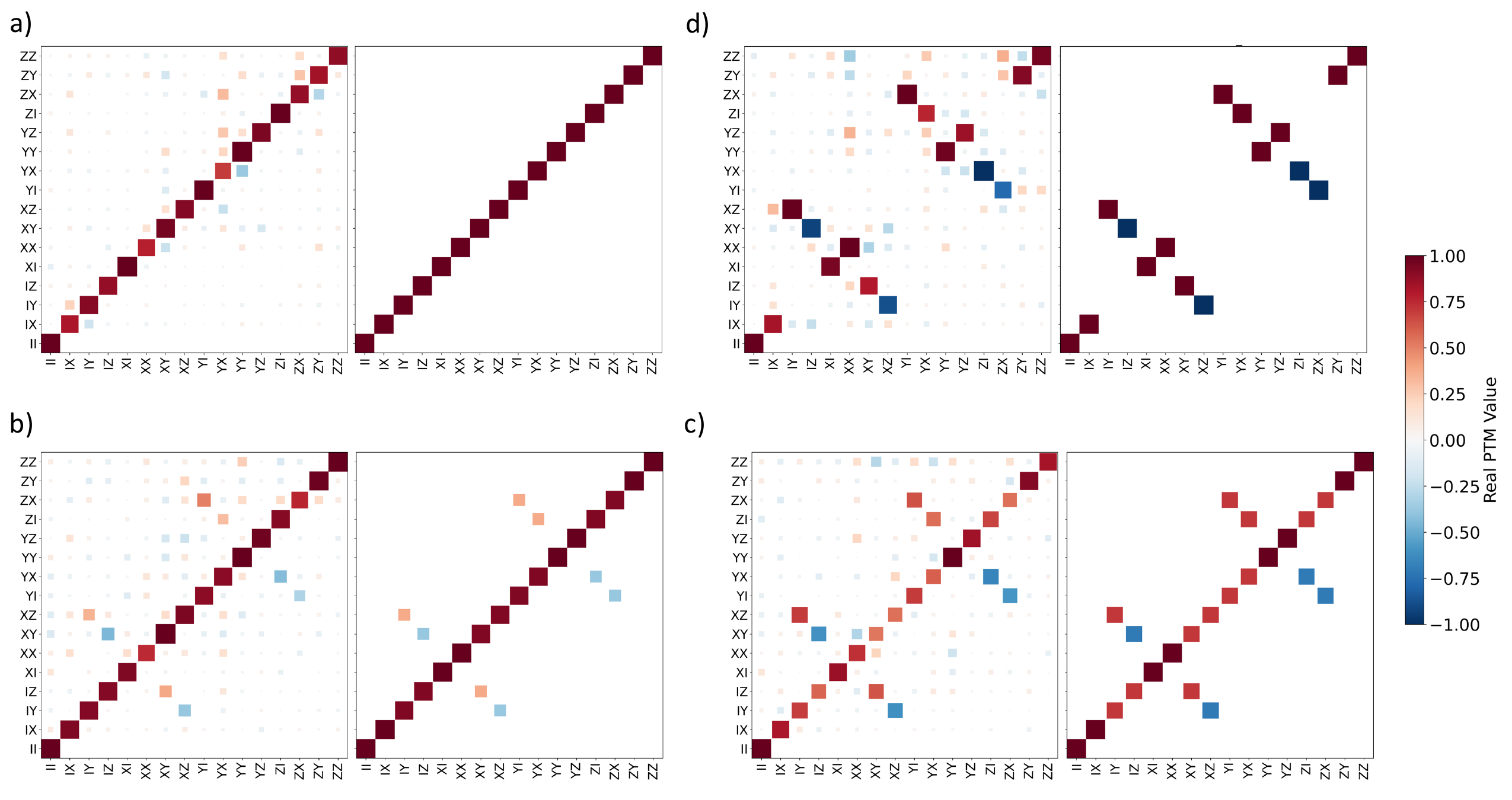}
    \caption{Measured and analytic Pauli transfer matrix (PTM) results of XX interaction on pair QA-QB at angle a) $\phi=0$, b) $\phi=\pi/8$, c) $\phi=\pi/4$ and d) $\phi=\pi/2$. The average process fidelity of each is a) $93.27\%$, b) $97.76\%$, c) $90.06\%$ and d) $94.14\%$.}
    \label{XX_AB}
\end{figure}

The extracted process fidelities, typically in the 90–99\% range, confirm that the implemented dynamics closely track the intended unitaries across evolution angles $\phi=0,\pi/8,\pi/4,\pi/2$. We further benchmark composite circuits and standard entangling operations such as CNOT and SWAP gates (Fig.~\ref{Circuit_A_AB_AD}), highlighting both their achievable fidelities and the cumulative errors arising in multi-step constructions. The measured PTM matrices were analyzed through post-selection on the measurement outcomes to mitigate the effect of readout errors, following the same approach as in Refs.~\cite{Cao_2024,Caoqtr2024}.

\begin{figure}[H]
    \centering
    \includegraphics[width=0.9\linewidth]{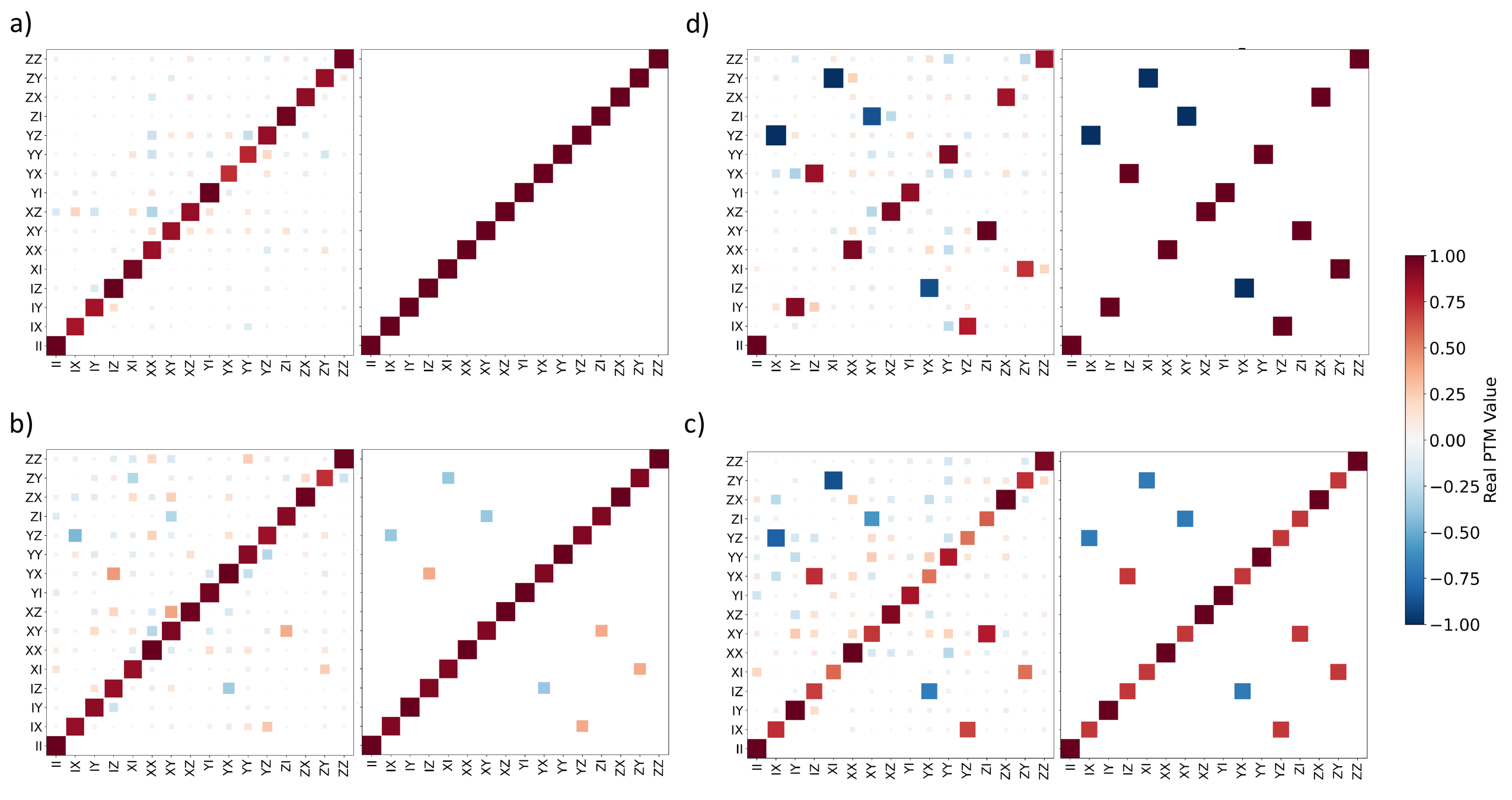}
    \caption{Measured and analytic Pauli transfer matrix (PTM) results of YY interaction on pair QA-QB at angle a) $\phi=0$, b) $\phi=\pi/8$, c) $\phi=\pi/4$ and d) $\phi=\pi/2$. The average process fidelity of each is a) $91.46\%$, b) $96.99\%$, c) $96.53\%$ and d) $93.03\%$.}
    \label{YY_AB}
\end{figure}

% \clearpage

% \newpage
% ZZ, XX and YY for QA-QD
\begin{figure}[H]
    \centering
    \includegraphics[width=0.92\linewidth]{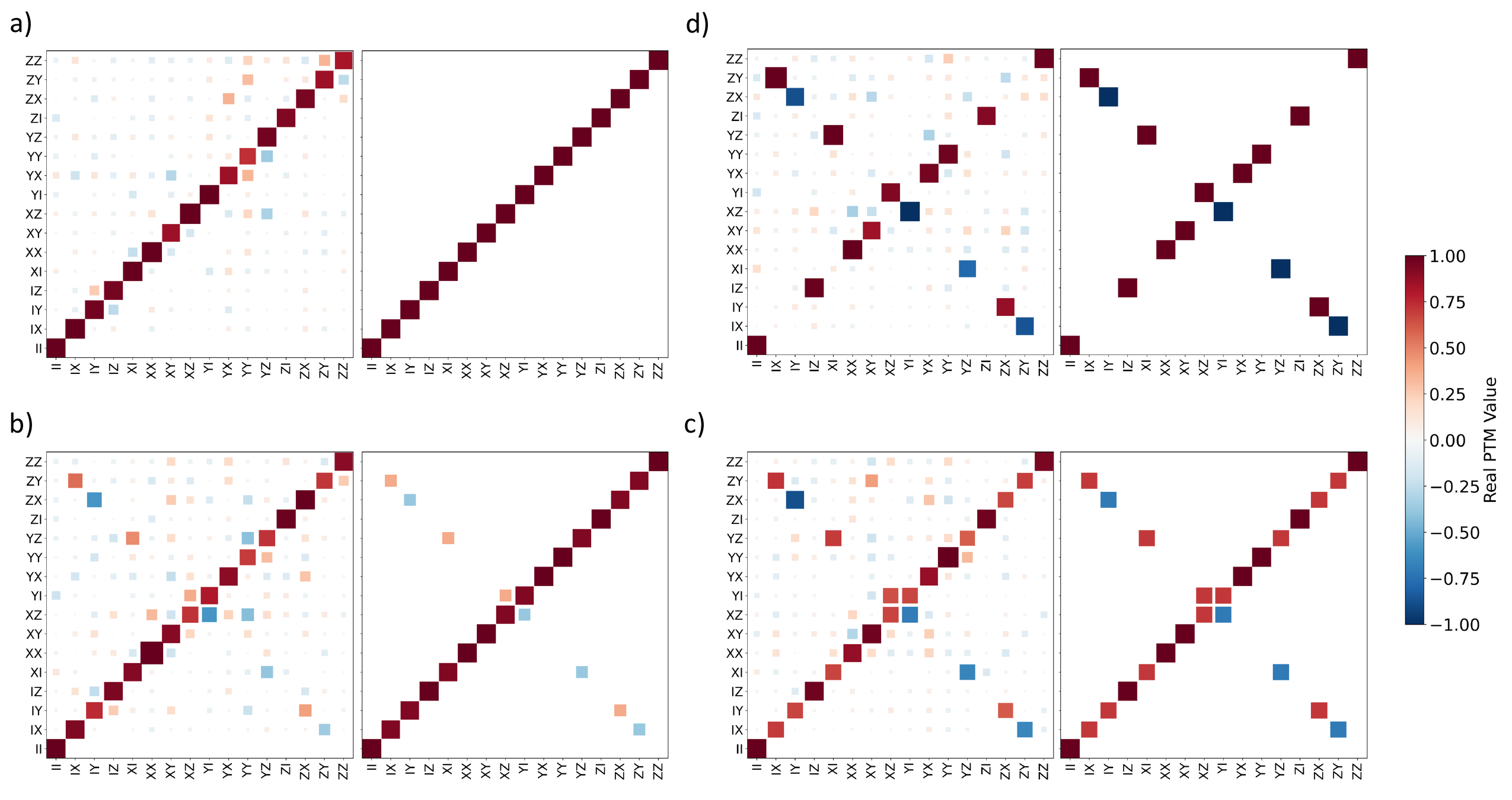}
    \caption{Measured and analytic Pauli transfer matrix (PTM) results of ZZ interaction on pair QA-QD at angle a) $\phi=0$, b) $\phi=\pi/8$, c) $\phi=\pi/4$ and d) $\phi=\pi/2$. The average process fidelity of each is a) $95.84\%$, b) $96.04\%$, c) $97.45\%$ and d) $96.74\%$.}
    \label{ZZ_AD}
\end{figure}

\begin{figure}[H]
    \centering
    \includegraphics[width=0.92\linewidth]{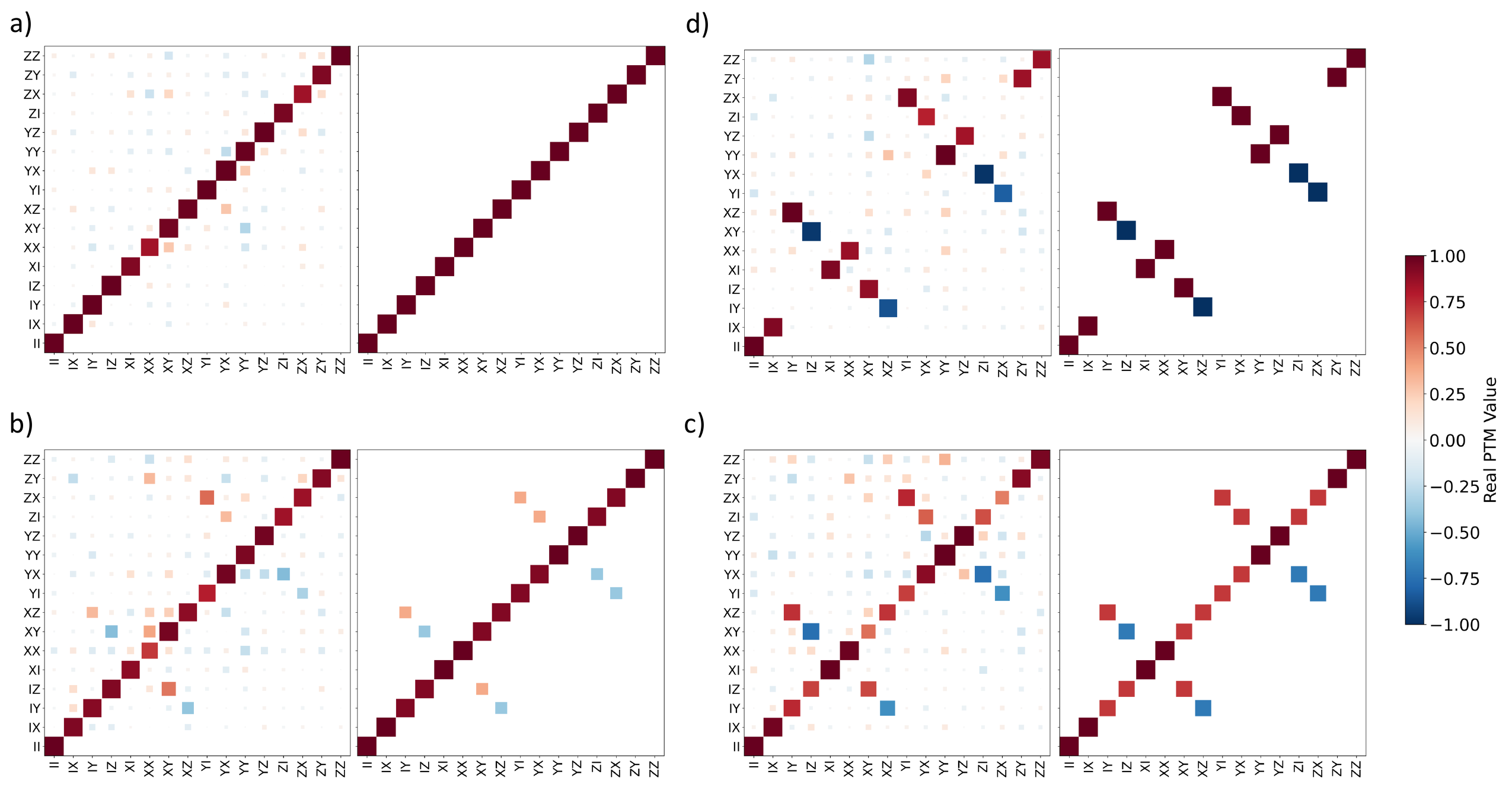}
    \caption{Measured and analytic Pauli transfer matrix (PTM) results of XX interaction on pair QA-QD at angle a) $\phi=0$, b) $\phi=\pi/8$, c) $\phi=\pi/4$ and d) $\phi=\pi/2$. The average process fidelity of each is a) $98.89\%$, b) $96.03\%$, c) $98.96\%$ and d) $93.12\%$.}
    \label{XX_AD}
\end{figure}

\begin{figure}[H]
    \centering
    \includegraphics[width=0.92\linewidth]{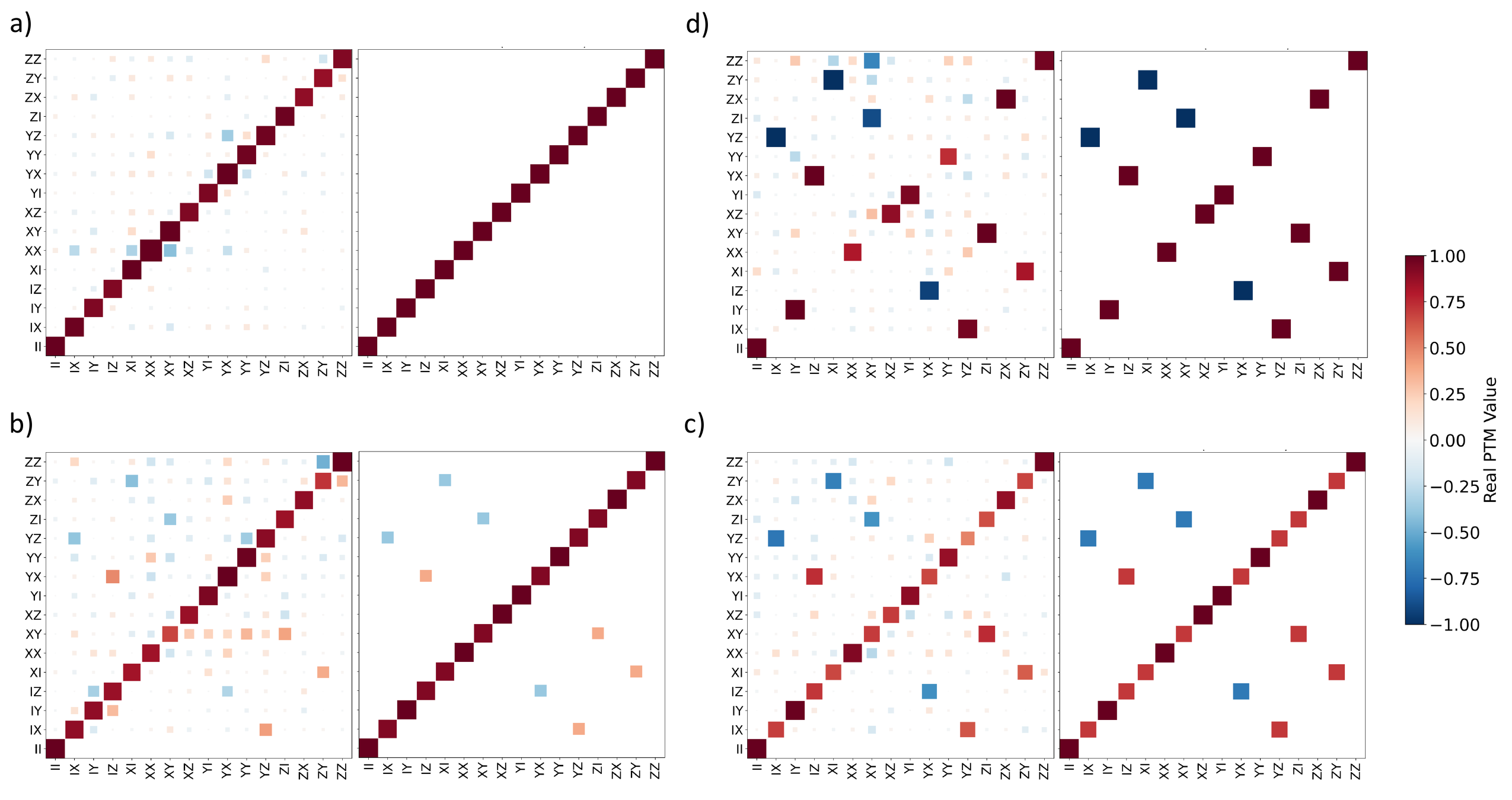}
    \caption{Measured and analytic Pauli transfer matrix (PTM) results of YY interaction on pair QA-QD at angle a) $\phi=0$, b) $\phi=\pi/8$, c) $\phi=\pi/4$ and d) $\phi=\pi/2$. The average process fidelity of each is a) $99.57\%$, b) $94.64\%$, c) $93.55\%$ and d) $95.68\%$.}
    \label{YY_AD}
\end{figure}

\clearpage

\newpage

\begin{figure}[H]
    \centering
    \includegraphics[width=1.0\linewidth]{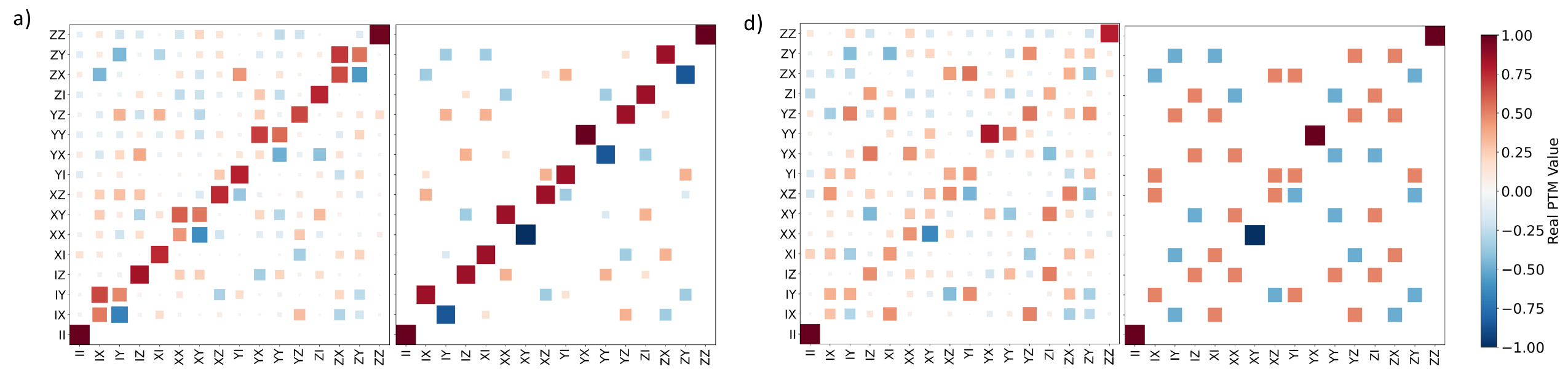}
    % \caption{Circuit A PTM Interactions on Pairs QA-QB and QA-QD}
    % \label{fig:xx}
    \caption{Measured and analytic Pauli transfer matrix (PTM) results for Circuit A angle a) $\phi=\pi/8$ and b) $\phi=\pi/4$. The average process fidelity of each is a) $86.11\%$ and b) $82.76\%$.}
    \label{Circuit_A_AB_AD}
\end{figure}

\begin{figure}[H]
    \centering
    \includegraphics[width=1.0\linewidth]{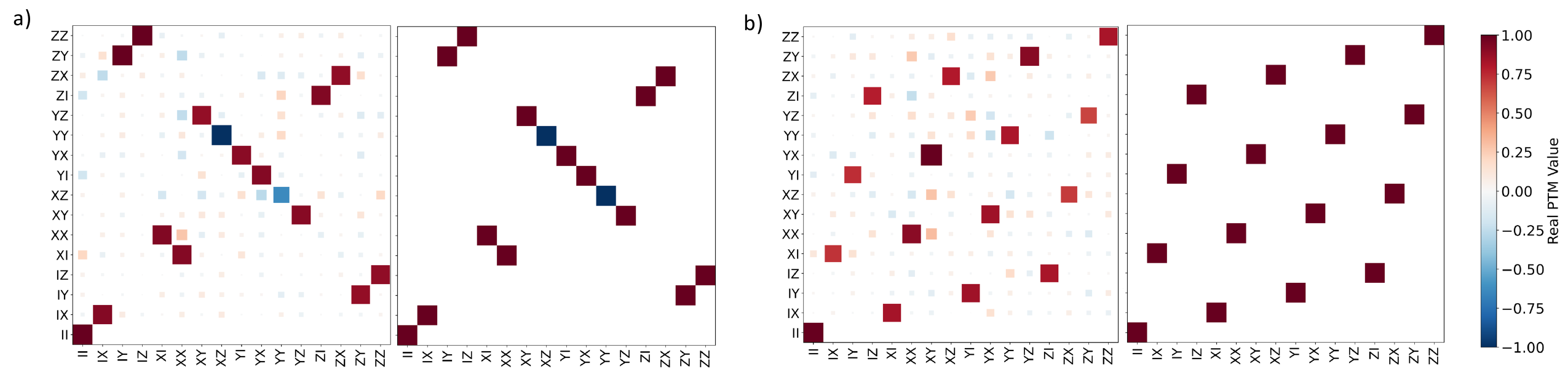}
    % \caption{CNOT and SWPA Gates on Pair QA-QD}
    % \label{fig:xx}
    \caption{Measured and analytic Pauli transfer matrix (PTM) results for a) CNOT gate and b) SWAP based on three CNOT gates. The average process fidelity of each is a) $93.05\%$ and b) $86.72\%$.}
    \label{Circuit_A_AB_AD}
\end{figure}

\subsubsection{DBAC energy measurements and full demonstration}

This section present the full measured energy trajectories obtained from applying DBAC circuits A–C in Table.\ref{tab:k_m_values_M1} at $\phi=\pi/4$, with initialization states parameterized as $\ket{\psi}=R_X(\theta)\ket{0}$. For each circuit, we extract the expectation values $E_i$ corresponding to the instruction qubit (the qubit subject to cooling) and the data qubits that provide the auxiliary resources. 

The sign of the extracted $E_i$ arises from our convention of expressing qubit energies in terms of expectation values of Pauli operators with eigenvalues $\pm 1$. In this representation, negative $E_i$ corresponds to population inversion or coherent phase accumulation relative to the chosen measurement axis, consistent with the oscillatory dynamics of the DBAC protocol. What matters for cooling is the relative reduction in the instruction qubit’s energy magnitude, balanced by a compensating increase in the data qubits, which reflects the redistribution of entropy across the system. Importantly, these negative values do not represent physical energies below the ground state, but rather the chosen operator convention for tracking coherent dynamics.

\begin{figure}[H]
    \centering
    \includegraphics[width=0.9\linewidth]{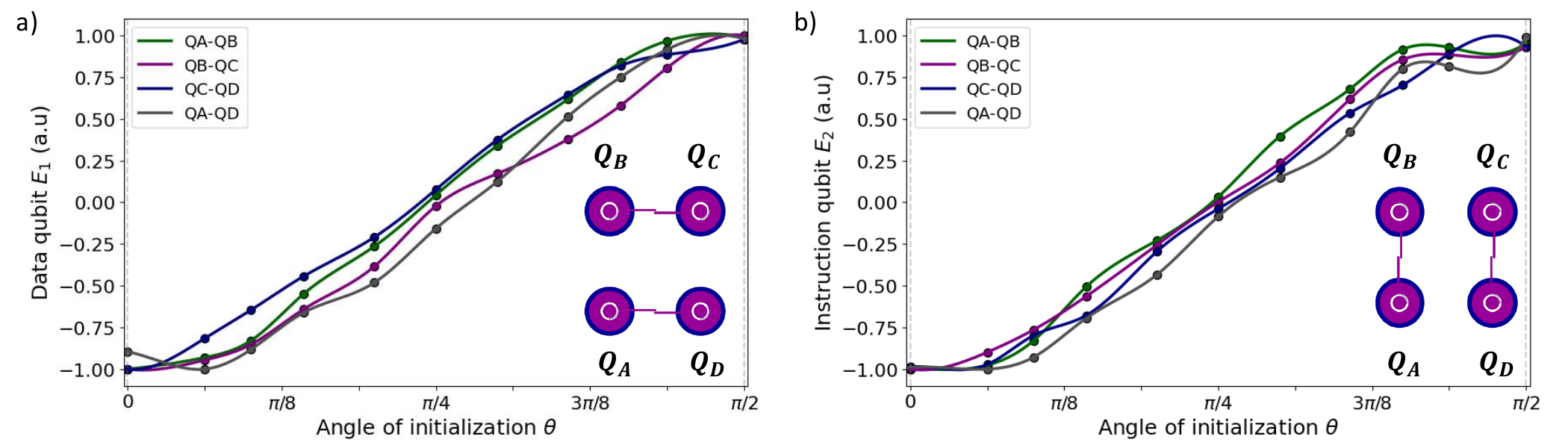}
    \caption{Circuit A: Measurements of energy cooling after applying DBAC circuits on qubit for $\phi=\pi/4$ and with different initialisation angle theta, where $\ket \psi = R_X(\theta)\ket 0$. $E_{1}$: green QA, purple QB, blue QC and gray QA. $E_{2}$: green QB, purple QC, blue QD and gray QD. $E_{1}$ for instruction qubit (main qubit to be cooled), and $E_{2}$ for data qubit.}
    \label{fig:enter-label}
\end{figure}
% \vspace{-2em}

Circuit A represents the minimal two-qubit implementation, where the instruction qubit exhibits a clear reduction in energy while the accompanying data qubit absorbs the excess, consistent with entropy redistribution. Circuit B extends the protocol to three qubits and demonstrates improved cooling performance, as the additional data qubit results in further cooling. Circuit C generalizes the protocol to four qubits, consistent with the anticipated recursive enhancement in cooling depth.

\begin{figure}[H]
    \centering
    \includegraphics[width=1\linewidth]{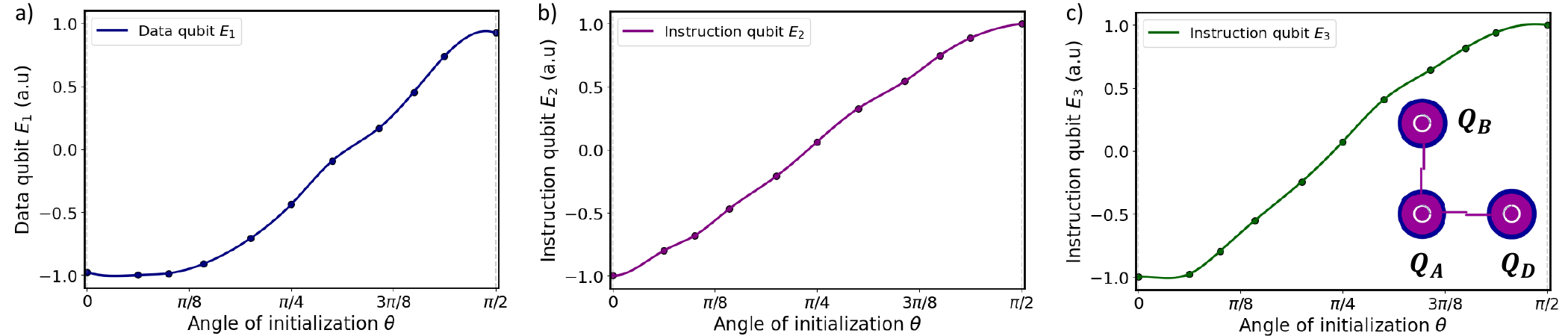}
    \caption{Circuit B: Measurements of energy cooling after applying DBAC circuits on qubit for $\phi=\pi/4$ and with different initialisation angle theta, where $\ket \psi = R_X(\theta)\ket 0$. $E_{1}$: blue QA, $E_{2}$: purple QB, and $E_{3}$: green QD. $E_{1}$ for instruction qubit (main qubit to be cooled), and $E_{2}$ and $E_{3}$ for data qubits.}
    \label{fig:enter-label}
\end{figure}
% \vspace{-2em}

 Across all three circuits, the dependence on initialization angle $\theta$ aligns with analytic predictions: qubits initialized closer to $\ket{0}$ converge rapidly to the ground state, while states initialized near $\ket{1}$ highlight the intrinsic performance limit of finite-step DBAC. These measurements constitute the full experimental demonstration of DBAC on our superconducting lattice, establishing its capability to coherently suppress excitations without tomography or measurement feedback.

\begin{figure}[H]
    \centering
    \includegraphics[width=0.8\linewidth]{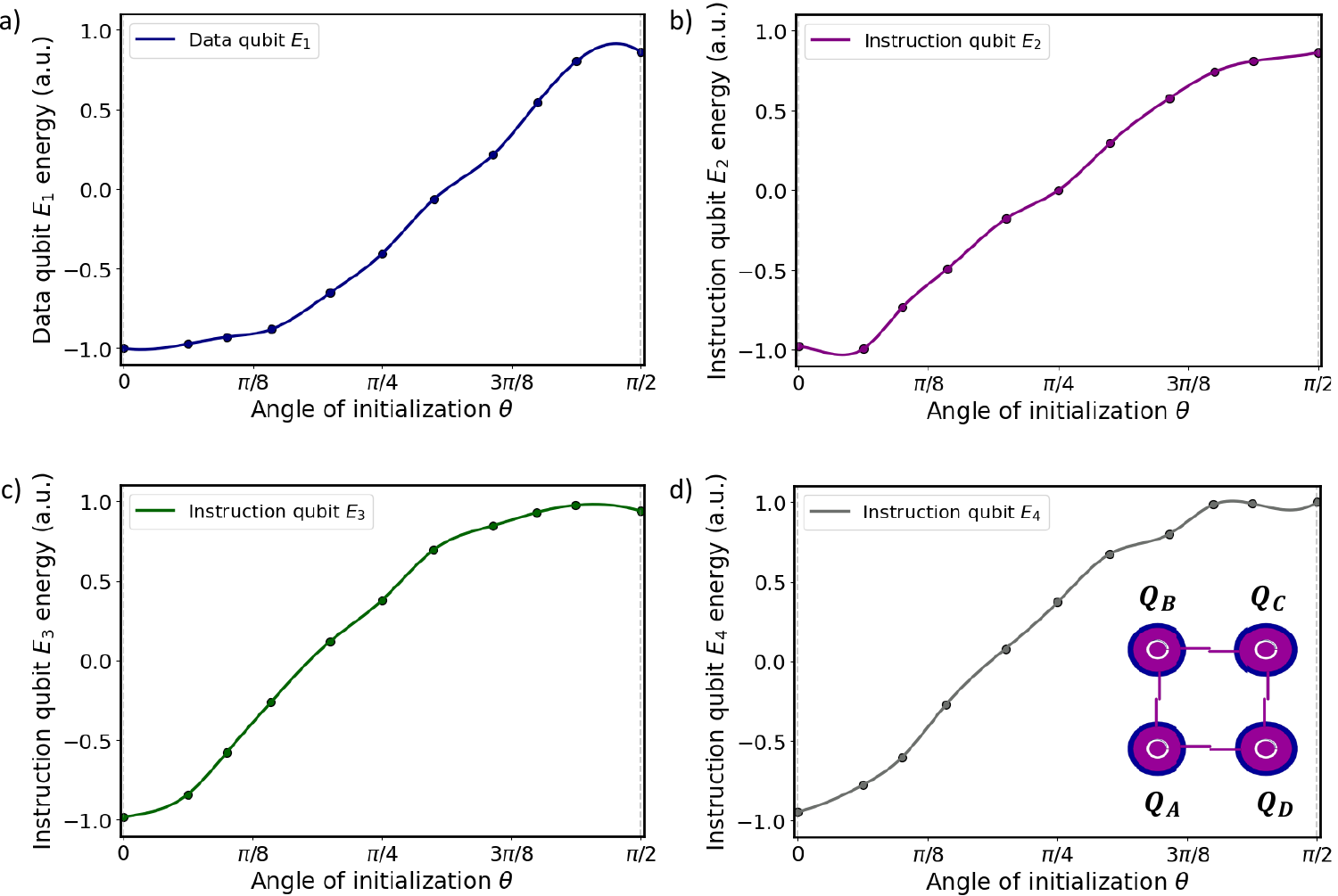}
    \caption{Circuit C: Measurements of energy cooling after applying DBAC circuits on qubit for $\phi=\pi/4$ and with different initialisation angle theta, where $\ket \psi = R_X(\theta)\ket 0$. $E_{1}$: blue QA, $E_{2}$: purple QB, $E_{3}$: green QC and $E_{4}$: gray QD. $E_{1}$ for instruction qubit (main qubit to be cooled), and $E_{2}$, $E_{3}$  and $E_{4}$ for data qubits.}
    \label{fig:enter-label}
\end{figure}
\clearpage
\onecolumngrid
\section{Additional discussion in context of other works}
\label{app:E}
\subsection{Comparison to heat-bath algorithmic cooling}

Algorithmic cooling (AC) was first introduced in \cite{schulman99} as a quantum counterpart analogue of the classical Carnot engine. 
The procedure was demonstrated in NMR systems, where qubits correspond to nuclear spins. 
The setting of AC involves an ensemble of $n$ identical qubits with initial polarization $\epsilon_0$, in other words, the qubits are initialised to mixed state:
\begin{align}
    \rho(\epsilon_0) = \frac{1}{2} \begin{pmatrix}
1 + \epsilon_0 & 0 \\
0 & 1 - \epsilon_0
\end{pmatrix}
\end{align}
The core idea of AC is to redistribute entropy from the target qubit to other qubits via a compression unitary $U_c$, and subsequently discard the computational qubits that receive entropy and heat up. 
In heat bath algorithmic cooling (HBAC) where there is a heat bath, instead of discarding the heated up computational qubits, we reinitialize the qubits through thermal contact with a bath typically biased toward the ground state with polarisation $\epsilon_b$~\cite{Park2016}.
The channel for HBAC is thus as follows:
\begin{align}
    \rho_{\epsilon_0}^{\otimes n} \to \text{Tr}_{2,3,\dots,n}(U_c\rho_{\epsilon_0}^{\otimes n} U_c^\dagger)\otimes\rho_{\epsilon_b}^{\otimes n-1}
\end{align}

This inspires the DBAC channel:
\begin{align}
    \rho_{\text{DBAC}}^{\otimes n} \to \text{Tr}_{2,3,\ldots, n}(U_D\rho_{\text{DBAC}}^{\otimes n}U_D^\dagger)
\end{align}
where $U_D$ comprises rounds of coupling between the target qubit and the instruction qubits and wraps around the Hamiltonian evolution in order to approximate the action of the unitary $U_{\text{DBAC}}= e^{it\hat{H}}e^{it\rho_{\text{DBAC}}}e^{-it\hat{H}}$ that defines a DBAC step.
Note, however that unlike in HBAC, where qubits are resetted to the initial thermal state, to perform a new DBAC step we need to use copies of the updated state, e.g. have $n^2$ initial states and perform the above DBAC step $n$ times in order to be able to obtain 1 copy of a state after 2 steps of DBAC.

Like HBAC, DBAC is also an iterative cooling algorithm, but it is designed for general systems that might lack access to a thermal bath. 
DBAC operates on copies of pure states and focuses on cooling a single qubit using a simple, hardware-native partial SWAP decomposition. 
In contrast, HBAC handles arbitrary mixed states, targets an entire ensemble, and relies on more complex unitary operations including full SWAP operations and Toffoli gates. 
Ultimately, DBAC is an approach distinct from HBAC, even though it shares the same setting and is likely implementable with the same (or weaker) control capabilities as in HBAC. 
It emphasizes practicality and low overhead in thermally isolated systems, whereas HBAC makes use of thermalization in open systems to achieve cooling.

\begin{figure}
    \centering
    \scalebox{1.0}{
\Qcircuit @C=1.0em @R=0.8em @!R { \\
	 	\nghost{{q}_{0} :  } & \lstick{\text{Target: } } & \qswap & \qw & \ctrl{2} & \targ & \ctrl{2} & \qw & \qw\\
	 	\nghost{{q}_{1} :  } & \lstick{} & \qw & \qswap & \targ & \ctrl{-1} & \targ & \qw & \qw\\
	 	\nghost{{q}_{2} :  } & \lstick{\text{Reset:}  } & \qswap \qwx[-2] & \qswap \qwx[-1] & \targ & \ctrl{-1} & \targ & \qw & \qw\\
\\ }}
    \caption{Algorithmic cooling circuit using PPA, where the entropy is compressed into the reset qubit. For HBAC, the reset qubit will come into contact with the heatbath and cool down.}
    \label{circuit:ac}
\end{figure}
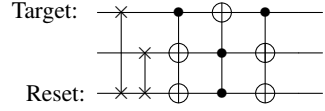

\subsection{Unitary synthesis methods for DME}
We presented an alternative experimental implementation of DME, differing from that of Ref.~\cite{kjaergaard} in that we used hardware-native gates more directly while Ref.~\cite{kjaergaard} used a decomposition to 3 CZ gates of Heisenberg interactions required for each DME step.
The approach of Ref.~\cite{kjaergaard} is better suited for fault-tolerant quantum computers as CZ gates are transversal for many error-correcting codes while our native gates are usually not.
On the other hand, for  experiments without error correction, it seems generally easier to use the native interactions as presented in the main text.

We note that Ref.~\cite{kjaergaard} demonstrated measurement emulation which  necessitates prior knowledge of the instruction state but then allows to perform more DME steps without requiring more copies of the instruction qubits.
For the task of removing coherence, this could be explored too but the ultimate performance would depend on the quality of the prior knowledge of the input states.

\subsection{Prospects of AC to reset qubits}

In this subsection we review resetting qubits, i.e., mapping registers in arbitrary states to always on state $\ket 0$.
This can be though of as the result of acting with an irreversible channel and indeed this is the standard approach to qubit reset.

For example, in Ref.~\cite{Geerlings_2013}, a double-drive reset protocol was applied to a transmon qubit in a 3D cavity, achieving ground state preparation with 99.5\% fidelity in under 3 $\mu s$. 
Similarly, ~\cite{Magnard_2018} presented a microwave-induced reset scheme for a transmon-resonator system, achieving less than 0.2\% residual excitation in 500 ns. 
These hardware-based methods provide fast and high-fidelity resets but are tailored to specific architectures.

AC protocols could be thought of as performing reset but the key difference is that for quantum computing application one wants to make rounds of operations on a single qubit, i.e. reset a single register and multiple copies of the same quantum states are not typically available.
In contrast, AC protocols require the input of $n$ copies of the same register which limits their applicability as a reset protocol for quantum computing.

Thus, for the task of resetting qubits DBAC, and all AC protocols, are constrained by the requirement for multiple copies of the same state, which is an unlikely scenario in real-world applications.
However, one could imagine that tailored methods might face a limitation occurring from some uncharacterized source of error.
If in some quantum computing application a reset of ultra-high quality would be required then one could imagine using that erroneous protocol to prepare $n$ qubits and employ AC to increase their quality.

We remark, that for resetting qubits based on multiple input qubits there are other known protocols which share the same setting with AC but are not considered explicitly as AC.
For example the CEM protocol, proposed by Cirac, Ekert, and Macchiavello ~\cite{Cirac1999_purification}, can improve the purity of a single qubit by operating on multiple degraded copies via symmetric subspace projection. 
Despite its significance, experimental realizations of the CEM protocol remain challenging due to the precision and control required. 
Recent NMR-based implementations~\cite{cemnmr} have validated its effectiveness, showing a purification factor of 1.25 under specific conditions. 
We remark that the CEM protocol can be generalized~\cite{Childs2025streamingquantum} which has been recently used as subroutine of a  implementing dynamic quantum algorithms in the framework of quantum dynamic programming~\cite{qdp}.

\bibliographystyle{quantum}
\bibliography{bibliography}

\end{document}